\documentclass[sigconf]{aamas}

\usepackage{balance} % for balancing columns on the final page
\usepackage{tcolorbox}
\setcopyright{ifaamas}
\acmConference[AAMAS '25]{Proc.\@ of the 24th International Conference
on Autonomous Agents and Multiagent Systems (AAMAS 2025)}{May 19 -- 23, 2025}
{Detroit, Michigan, USA}{A.~El~Fallah~Seghrouchni, Y.~Vorobeychik, S.~Das, A.~Nowe (eds.)}
\copyrightyear{2025}
\acmYear{2025}
\acmDOI{}
\acmPrice{}
\acmISBN{}

\acmSubmissionID{1010}

\usepackage{xcolor, dsfont}
\usepackage{amsmath}
\usepackage{algorithm}
\usepackage{algorithmic}

\usepackage{epsfig, amsmath,epsf,scalefnt,epstopdf,setspace,fancyhdr,enumerate,mathrsfs, array}
\usepackage{xcolor}

\allowdisplaybreaks[4]

\usepackage{todonotes}
\usepackage{multirow}
\usepackage{booktabs}
\usepackage{tabularx}
\usepackage[htt]{hyphenat}

\usepackage{graphicx}     % For handling images (dummy image below)

\usepackage{tikz}
\usetikzlibrary{automata,arrows,positioning,calc}

\newtheorem{thm}{Theorem}

\newtheorem{prop}{Proposition}
\newtheorem{cor}{Corollary}
\newtheorem{defn}{Definition}

\usepackage{amsfonts,comment,bbm}

\newcommand{\cF}{\mathcal{F}}

\newcommand{\cA}{\mathcal{A}}
\newcommand{\cS}{\mathcal{S}}
\newcommand{\ust}{^{\star}}

\newcommand{\cI}{\mathcal{I}}

\newtheorem{lemma}{Lemma}

\newtheorem{remark}{Remark}

\newtheorem{example}{Example}

\def\cA{{\mathcal{A}}}   
 \def\cF{{\mathcal{F}}}  
\def\cI{{\mathcal{I}}}   
   
  \def\cS{{\mathcal{S}}}

   \def\bs{{\mathbf{s}}}

\title[AAMAS-2025 Formatting Instructions]{Finite-Horizon Single-Pull Restless Bandits: An Efficient Index Policy For Scarce Resource Allocation}

\author{Guojun Xiong$^1$, Haichuan Wang$^1$, Yuqi Pan$^1$, Saptarshi Mandal$^2$, Sanket Shah$^1$, Niclas Boehmer$^3$, Milind Tambe$^1$  }
\affiliation{
  \institution{$^1$Harvard University, $^2$University of Illinois, Urbana- Champaign, $^3$Hasso Plattner Institute}
  \city{$^1$Cambridge, $^2$Champaign, $^3$Potsdam}
  \country{$^{1,2}$United States, $^3$Germany}}
\email{{gjxiong,
 haichuan_wang,  yuqipan,sanketshah, tambe}@g.harvard.edu, smandal4@illinois.edu,niclas.boehmer@hpi.de}

\begin{abstract}
Restless multi-armed bandits (RMABs) have been highly successful in optimizing sequential resource allocation across many domains. However, in many practical settings with highly scarce resources, where each agent can only receive at most one resource, such as healthcare intervention programs, the standard RMAB framework falls short. To tackle such scenarios, we introduce Finite-Horizon Single-Pull RMABs (SPRMABs), a novel variant in which each arm can only be pulled once. This single-pull constraint introduces additional complexity, rendering many existing RMAB solutions suboptimal or ineffective. %To address this, we propose using dummy states to duplicate the system, ensuring that once an arm is activated, it transitions exclusively within the dummy states.
To address this shortcoming, we propose using \textit{dummy states} that expand the system and enforce the one-pull constraint. We then design a  lightweight index policy for this expanded system.  For the first time, we demonstrate that our index policy achieves a sub-linearly decaying average optimality gap of
$\tilde{\mathcal{O}}\left(\frac{1}{\rho^{1/2}}\right)$ for a finite number of arms, where
$\rho$ is the scaling factor for each arm cluster. Extensive simulations validate the proposed method, showing robust performance across various domains compared to existing benchmarks.
\end{abstract}

\keywords{Restless multi-armed bandits, scare resource, single-pull constraint, index policy, optimality analysis}

\newcommand{\BibTeX}{\rm B\kern-.05em{\sc i\kern-.025em b}\kern-.08em\TeX}

\begin{document}

%%% The following commands remove the headers in your paper. For final
%%% papers, these will be inserted during the pagination process.

\pagestyle{fancy}
\fancyhead{}

%%% The next command prints the information defined in the preamble.

\maketitle

%%%%%%%%%%%%%%%%%%%%%%%%%%%%%%%%%%%%%%%%%%%%%%%%%%%%%%%%%%%%%%%%%%%%%%%%
\section{Introduction}
The restless multi-armed bandits (RMAB) problem \citep{whittle1988restless} is a time slotted game between a decision maker (DM) and the environment.  In the standard RMAB model, each ``restless'' arm is described by a Markov decision process (MDP) \citep{puterman1994markov}, and evolves stochastically according to two different transition functions, depending on whether the arm is activated or not. Scalar rewards are generated with each transition. The goal of the DM is to maximize the total expected reward under an instantaneous constraint that at most $K$ out of $N$ arms can be activated at any decision epoch.  RMABs have been widely used to model a variety of real-world applications such as problems around congestion control \cite{avrachenkov2017controlling}, job scheduling \cite{yu2018deadline},  wireless communication \cite{bagheri2015restless}, healthcare \cite{mate2021risk,killian2021beyond}, queueing systems \cite{larrnaaga2016dynamic}, and cloud computing \cite{xiong2022reinforcementcache}. One key reason for the popularity of RMABs is their ability to optimize sequential allocation of limited resources to a population of agents in uncertain environments \cite{mate2022field}.

% In many of real-world domains ranging from healthcare, to conservation, to machine maintenance, however, the number of resources are very limited, so each agent may get a resource only once, i.e., each arm in an RMAB may be pulled only once. In public health, RMABs are used to optimize allocation of health intervention resources, and each person may get an intervention only once. Concrete examples include maternal health in India, where RMABs are deployed to allocate limited resources of live service calls from a call center to ensure mothers do not drop off from important health programs. Each RMAB arm models one mother, where once one arm is pulled – i.e., the mother receives a service call – there is a very long gap required before it can be pulled again\cite{verma2023expanding}, essentially implying that each arm is pulled only once. Similarly in the case of maternal health in Uganda\cite{,}, where RMABs are used to model allocation of scarce wireless vital sign monitors to mothers in maternity wards; a mother may receive such a monitor only once during her stay in the maternity ward.  We may similarly see that in combating malnutrition, a child may be enrolled in a malnutrition program only once....
% %{\color{blue}What healthcare examples can motivate the one pull constraint well? 1. IAAI paper on arxiv. 2. Example from Armman, dataset from Sanket and shreth, how many times each arm is pulled}

However, in many real-world scenarios, additional constraints are placed on the allocation. In this paper, we propose and study the new problem of sequentially allocating resources when each agent can only receive a resource in at most one timestep, i.e., we focus on the RMAB problem where no arm can be pulled repeatedly. This constraint is, for instance, prevalent in real-world domains where resources are extremely scarce and there are many more agents than resources, occurring for instance in healthcare, conservation, and machine maintenance. Even in cases where the number of resources and agents are of the same magnitude, organizers might impose single-pull constraints for fairness reasons to ensure an equal treatment of all agents. Lastly, there are also allocation scenarios where an agent only benefits from the first resource assigned to them, for instance, when distributing single-dose vaccines.

Concrete examples where the single-pull constraint is imposed in practice arise in public health, where RMABs are (or can be) used to optimize the allocation of health intervention resources~\cite{lee2019optimal,verma2023expanding,boehmer2024optimizingvitalsignmonitoring}. We present more detailed examples of deployed and emerging applications from healthcare and other domains with the presence of such a single intervention constraint in Section \ref{Motivating}. These practical challenges necessitate the development of a new model capable of addressing the unique constraints posed by single-pull scenarios, ensuring efficient allocation of limited resources.

%However, in many real-world domains, such as healthcare, conservation, and machine maintenance, resources are often extremely limited, meaning that each agent may receive a resource only once, i.e., an RMAB arm may be pulled only once. For example, in public health, RMABs are used to optimize the allocation of health intervention resources~\cite{lee2019optimal,verma2023expanding,boehmer2024optimizingvitalsignmonitoring}, where each person may receive an intervention only once. We present more detailed examples of deployed and emerging applications from healthcare and other domains with the presence of such a single intervention constraint, i.e., single-pull constraint, in Section \ref{Motivating}. These practical challenges necessitate the development of a new model capable of addressing the unique constraints posed by single-pull scenarios, ensuring efficient  allocation of limited resources.

Given the new and urgent requirement of a single pull per arm in many practical domains, we introduce the Finite-Horizon Single-Pull Restless Multi-Armed Bandits (\texttt{SPRMABs}), a novel variant of the RMAB framework where each arm can be pulled at most once.  As widely known, the complexity of conventional RMAB lies in the challenge of finding an optimal control strategy to maximize the expected total reward, a problem that is typically intractable \cite{papadimitriou1999complexity}. As a result, existing approaches have largely focused on designing efficient heuristic (and at times asymptotically optimal) index-based policies, such as those developed for offline RMABs \citep{whittle1988restless, larranaga2014index, bagheri2015restless, verloop2016asymptotically, zhang2021restless} and reinforcement learning (RL) algorithms for online RMABs \citep{tekin2011adaptive, cohen2014restless, fu2019towards, jung2019regret,  nakhleh2021neurwin, xiong2022Nips, avrachenkov2022whittle, xiong2023finite, wang2024online}. However, the introduction of the single-pull constraint in \texttt{SPRMABs} renders these traditional methods either suboptimal or inapplicable.

Tailoring existing methods, such as the Whittle index policy \cite{whittle1988restless} and fluid linear programming (LP)-based policies \cite{verloop2016asymptotically, zhang2021restless, brown2023fluid, gast2023linear}, to the novel \texttt{SPRMABs} presents significant challenges. Specifically, the Whittle index policy \cite{whittle1988restless} is defined using Lagrange multipliers for activation budget constraints. Introducing the single-pull constraint disrupts this framework, as the Lagrangian formulation becomes ill-defined, causing the method to return highly suboptimal solutions.
%fail entirely. 
Similarly, for fluid LP-based methods \cite{verloop2016asymptotically, zhang2021restless, brown2023fluid, gast2023linear}, enforcing the single-pull constraint introduces a new nonlinear constraint, which exponentially increases the complexity as the time horizon and the number of arms grows. The question we tackle in this paper is the following:

\vspace{-0.13in}
\begin{tcolorbox}[colback=white!5!white,colframe=white!75!white]
\textit{Is it possible to design a light-weight asymptotically optimal index policy for \texttt{SPRMABs}?}
\end{tcolorbox}
\vspace{-0.15in}
To tackle this challenge, we utilize \textit{dummy states} to duplicate the entire system, ensuring that once an arm is activated, it transitions exclusively between these dummy states. The transitions within the dummy states mirror those of the normal states when no action is taken (i.e., action 0). \textit{Building on this expanded system, we design a lightweight index policy specifically tailored for SPRMABs}, and we demonstrate that our proposed index policy achieves a linear decaying rate in the average optimality gap. Our main contributions can be summarized as follows:

$\bullet$ \textbf{Lightweight Index Policy Design:} We leverage the concept of expanding the system through dummy states and develop a lightweight index policy, called single-pull index (\texttt{SPI}) policy, which addresses two challenges that conventional index policies cannot handle. 
{First, in real-world applications, pulling an arm doesn't always guarantee better outcomes, meaning the full budget may not need to be used. Existing index policies often exhaust the budget on the highest indices, leading to suboptimal results. Second, in \texttt{SPRMABs}, each arm can only be pulled once, making activation timing crucial. An arm with a high index now may yield a better reward if pulled later, which current algorithms fail to handle. Dummy states allow deferring decisions without affecting future rewards, conserving resources and tackling both challenges effectively.}

$\bullet$ \textbf{Optimality Gap:} For the first time, we demonstrate that our proposed index policy achieves a sub-linearly decaying rate of the average optimality gap for a finite number of arms, characterized by the bound 
$\tilde{\mathcal{O}}(\frac{1}{\rho^{1/2}}+\frac{1}{\rho^{3/2}})$, where 
$\rho$ denotes the scaling factor for each arm cluster.

$\bullet$ \textbf{Empirical Simulations:} We conduct extensive simulations to validate the effectiveness of the proposed method, benchmarking it against existing strategies. The results consistently demonstrate robust performance across a variety of domain settings, underscoring the practicality and versatility of our index policy in addressing \texttt{SPRMABs}. This advancement not only enhances the applicability of \texttt{SPRMABs} in equitable resource allocation but also lays a strong foundation for future research in constrained bandit settings.

\section{Motivating Domains and Examples}
\label{Motivating}
The single-pull constraint in RMABs is motivated by multiple real-world domains. We begin by describing examples from public health domains with limited resources \cite{lee2019optimal,ayer2019prioritizing}.  One concrete deployed example RMABs used for a maternal mHealth (mobile health) program in India \cite{mate2022field, verma2023expanding}. This deployment supports an mHealth program of ARMMAN (armman.org), an India-based non-profit that spreads preventative care awareness to pregnant women and new mothers through an automated call service. To reduce dropoffs from the mHealth program, ARMMAN employs health workers to provide live service calls to beneficiaries; however, ARMMAN is faced with a resource allocation challenge because any one time, there are 200K beneficiaries (mothers) enrolled in the program but they have enough staff to only do 1000 live source calls per week. As a result, RMABs are deployed to optimize allocation of their limited live service calls\cite{verma2023expanding}, and given the scale of the program, a beneficiary received a maximum of one service call. That is,  each RMAB arm represents a mother, and once an arm is pulled, i.e., the mother receives a service call, she does not receive a service call again. 

Similarly, in maternal health programs in Uganda \cite{boehmer2024optimizingvitalsignmonitoring}, RMABs are proposed to be used to allocate scarce wireless vital sign monitors to mothers in maternity wards, where each mother may receive such a monitor only once during her stay. A similar scenario occurs in support programs which can only support a limited number of beneficiaries every week and beneficiaries can only participate in the program once. One such example is  malnutrition prevention\cite{langendorf2014preventing}, where a child may be enrolled in a malnutrition program only once. These practical challenges necessitate the development of a new model capable of addressing the unique constraints posed by single-pull scenarios, ensuring an efficient 
%and equitable 
allocation of limited resources.

\section{System Model and Problem Formulation}
Consider a finite-horizon RMAB problem with $N$ arms. Each arm $n$ is associated with a specific Unichain Markov decision process (MDP) $(\cS, \cA, P_n, r_n, \bold{s}_1, T)$, where $\cS$ is the finite state space and $\cA:=\{0,1\}$ denotes the binary action set. Using the standard terminology from the RMAB literature, we call an arm \textit{passive} when action $a=0$ is applied to it, and \textit{active} otherwise.  $P_n:\cS\times\cA\times\cS\mapsto\mathbb{R}$ is the transition kernel and $r_n: \cS\times\cA\mapsto \mathbb{R}$ is the reward function.
The total number of activated arms at each time $t$ is constrained by $K$, which we call the \textit{activation budget.} The initial state is chosen according to the initial distribution $\bold{s}_1$ and $T<\infty$ is the horizon.

At time $t\in[T]$, each arm $n$ is at a specific state $s_n(t)\in \cS$ and evolves to $s_n(t+1)$ independently as a controlled Markov process with the controlled transition probabilities $P_n(s_n(t),a_n(t), s_n(t+1))$ when action $a_n(t)$ is taken.
The immediate reward earned from activating arm $n$ at time $t$ is denoted by
$r_n(t) : = r_n(s_n(t), a_n(t))$.   Denote the total reward earned at time $t$ by $R(t)$, i.e., $R(t):= \sum_n r_n(t)$.
Motivated by the healthcare implementations where each arm (i.e., patient) can only be pulled for once due to resource limitation,
now let us consider the scenario where each arm can only be pulled once, and the duration of activation is also one. This is equivalent to the constraint in the following expression
\begin{align}\label{eq:one-pull-cons}
 \texttt{Single-pull constraint}:~~~~~~~   \sum_{t=1}^T a_n(t)\leq 1, \forall n.
\end{align}

Let $\cF_t$ denote the operational history until $t$, i.e., the {$\sigma$-algebra}  generated by the random variables $ \{s_n(\ell):n\in [N],\ell \in [t] \}, \{a_n(\ell):n\in [N],\ell \in [t-1] \}$. Our goal is to derive a policy $\pi: \cF_t \mapsto \cA^{N}$ that makes decisions regarding which set of arms are made active at each time $t\in[T]$ so as to maximize the expected value of the cumulative rewards subject to the activation budget and the one-pull constraint in \eqref{eq:one-pull-cons}, i.e.,
\begin{align}\label{eq:One-Pull-Opt}\nonumber
\texttt{SPRMAB}:~~&\max_{\pi}\mathbb{E}_\pi \left(\sum_{n=1}^{N} \sum_{t=1}^{T} r_n(t)\right)\\
&\quad\mbox{s.t.}\sum_{n=1}^{N}\! a_n(t)\leq K,\forall t\in [T], ~\sum_{t=1}^T a_n(t)\leq 1, \forall n.
\end{align}
where the subscript indicates that the expectation is taken with respect to the measure induced by the policy $\pi.$  We refer to the problem \eqref{eq:One-Pull-Opt} as the ``original problem'', which
%It is clear that the original problem
suffers from the ``curse of dimensionality'', and hence is computationally intractable. %{\color{red}[More details?]}.
We overcome this difficulty by developing a computationally feasible and provably optimal index-based policy.

\subsection{Existing Index Policies and Failure Examples}

The challenge comes from the ``hard" constraints in \eqref{eq:One-Pull-Opt}, where the first \texttt{budget constraint}  must be satisfied at every time step, and the second \texttt{single-pull constraint}  must be satisfied firmly for all arms.
Existing index policy approaches \cite{whittle1988restless,verloop2016asymptotically,zhang2021restless} for conventional RMAB problems without the \texttt{single-pull} constraint
design indices
by relaxing the ``hard'' \texttt{activation-budget constraint}  $\sum_{n=1}^{N}\! a_n(t)\leq K,\forall t\in [T]$ to the ``relaxed'' constraints, i.e., the activation cost at time $t\in[T]$ is limited by $K$ in expectation, which is
\begin{align}\label{eq:relaxed_budget_cons}
\texttt{Re-budget constraint:}~~~\mathbb{E}_{\pi}\left\{\sum_{n=1}^{N} a_n(t) \right\}\leq K.
\end{align}
In the following, we present two typical index polices, one is the Whittle index policy \cite{whittle1988restless}, and the other is the LP-based index policy \cite{verloop2016asymptotically,zhang2021restless,xiong2021reinforcement}.

\textbf{Whittle Index Policy.}  Whittle index \cite{whittle1988restless} is designed upon the infinite-horizon average-reward (IHAR) RMAB settings through decomposition.  Specifically, Whittle relies on the  \texttt{Relaxed budget constraint}  in~(\ref{eq:relaxed_budget_cons}) and obtains a unconstrained problem  for IHAR settings:
\begin{align*}
  \texttt{IHAR-RMAB:}~  \max_{\pi\in\Pi} \liminf_{T\rightarrow\infty}\frac{1}{T}\mathbb{E}_\pi\sum_{t=1}^T\sum_{n=1}^{N} \{r_n(t)+\lambda(1- a_n(t))\},
\end{align*} {where $\lambda$ is the Lagrangian multiplier associated with the constraint.}  The key observation of Whittle is that this problem can be decomposed and its solution is obtained by combining solutions of $N$ independent problems via solving the associated dynamic programming (DP) \cite{xiong2023finite}: $V_n(s)=\max_{a\in\{0,1\}}Q_n(s,a), \forall n\in\mathcal{N},$ where
\begin{align}\label{eq:Q_value}
Q_n(s,a)+{\beta}&=a\Big(r_n(s,a)+\sum_{s^\prime}p_n(s^\prime |s,1)V_n(s^\prime)\!\Big)\nonumber\\
&+(1-a)\Big(r_n(s,a)+\lambda+\sum_{s^\prime}p_n(s^\prime |s,0)V_n(s^\prime)\Big),
\end{align}
\vspace{-0.01in}
where $\beta$ is unique and equals to the maximal long-term average reward of the unichain MDP, and $V_n(s)$ is unique up to an additive constant, both of which depend on the Lagrangian multiplier $\lambda.$ The optimal decision $a^*$ in state $s$ then is the one which maximizes the right hand side of the above DP. The Whittle index associated with state $s$ is defined as the value $\lambda_n^*(s)\in\mathbb{R}$ such that actions  $0$ and $1$ are equally favorable in state $s$ for arm $n$ \cite{avrachenkov2022whittle,fu2019towards}, satisfying
\begin{align}\label{eq:Whittle_index}
\lambda_n^*(s) &:= r_n(s,1)+\sum_{s^\prime}p_n(s^\prime |s,1)V_n(s^\prime)\nonumber\\
&\qquad\qquad\qquad-r_n(s,0)-\sum_{s^\prime}p_n(s^\prime |s,0)V_n(s^\prime).
\end{align}
Whittle index policy then activates $K$ arms with the largest Whittle indices at each time slot $t$.

\textbf{LP-based Index Policy.}
With the \texttt{Relaxed budget constraint} in \eqref{eq:relaxed_budget_cons}, we can transfer the conventional RMAB problem into an equivalent LP \cite{altman2021constrained} by leveraging the definition of occupancy measure (OM). In particular, the OM $\mu$ of a policy $\pi$ in a finite-horizon MDP is defined as the expected number of visits to a state-action pair $(s,a)$ at each time $t$, i.e.,
$$
\mu= \left\{\mu_n(s,a;t)=\mathbb{P}(s_n(t)=s, a_n(t)=a): \forall n, t|0\leq \mu_n(s,a;t)\leq 1\right\},$$
which is a probability measure,  satisfing $\sum_{s,a}\mu_n(s,a,t)=1$, $
\forall t\in[T]$. Hence, the associated LP is expressed as
\begin{align}
\max_{\mu}&\sum_{n=1}^{N}\sum_{t=1}^T\sum_{(s,a)} \mu_n(s,a;t){r}_n(s,a)\label{eq:obj_rel_lp} \displaybreak[0]\\
\hspace{-0.1cm}\mbox{ s.t. }& { \sum_{n=1}^{N}\sum_{s} \mu_n(s,1;t)\leq K}, ~\forall t,{//activation ~constraint }
\label{eq:con_rel_lp}\displaybreak[1]\\
& {\sum_{a}} \mu_n(s,a;t)\!=\!\!\!\!\sum_{(s^\prime,a^\prime)}\!\!\!\mu_n(s^\prime, a^\prime; t\!-\!1)P_n(s^\prime, a^\prime,s),\forall n, s,\nonumber\\
&\qquad\qquad\qquad\qquad\qquad\qquad{//fluid ~equilibrium }\label{eq:pmc_rel_lp1}\displaybreak[3]\\
&~~\sum_{a}\mu_n(s,a;1)={\bold{s}_1(s)},~\forall s, n,{//initial ~distribution }
\label{eq:pmc_rel_lp2}
\end{align}
where \eqref{eq:con_rel_lp} is a restatement of the \texttt{budget constraint} in \eqref{eq:One-Pull-Opt} for $\forall t\in[T]$, which indicates the activation budget;
\eqref{eq:pmc_rel_lp1} represents the transition of the occupancy measure from time  $t-1$ to time $t$, $\forall n\in [N]$ and $\forall t\in[T]$; and \eqref{eq:pmc_rel_lp2} indicates the initial condition for occupancy measure at time $1$, $\forall s\in\cS$.
Denote the solution to the above LP as $\mu\ust = \left\{\mu\ust_{n}(s,a;t):n\in [N], t\in [T]  \right\}$.
A simple index-based policy according to the optimal solution $\mu^\star$ can be designed by dividing the arms at each time slot $t$ into three categories:
\begin{enumerate}
    \item High-priority states: $\mu_n^\star(s,0;t)=0$. (\textit{Pull arms under those states.} )
    \item Medium-priority states: $\mu_n^\star(s,1;t)>0$, $\mu_n^\star(s,0;t)>0$.(\textit{Pull arms under those states when remaining budget is available.})
    \item Low-priority states:
    $\mu_n^\star(s,1;t)=0$. (\textit{Do not pull arms under those states.})
\end{enumerate}

In \texttt{SPRMAB} as shown in \eqref{eq:One-Pull-Opt} where each arm can be activated at most once, the standard Whittle and LP-based index policies may become suboptimal. This is because the these index policies are designed under the assumption that arms can be activated multiple times, and it may not adequately account for the urgency of activating certain arms in a single-pull setting. Below, we present a rigorous example demonstrating how the Whittle and LP-based index policies can fail under these constraints (see Figure \ref{fig:example}).

% \begin{figure}
%     \centering
%     \resizebox{0.5\textwidth}{!}{\begin{tikzpicture}[->, >=stealth', shorten >=1pt, auto, thick, node distance=3cm]
%     \tikzstyle{state} = [draw, circle, minimum size=2cm, text centered]

%     % Nodes (States)
%     \node[state, fill=black, text=white] (NE) {NE};
%     \node[state, fill=gray!50] (E) [right of=NE] {E};
%     \node[state, fill=blue!30] (SE) [right of=E] {SE};

%     % Arrows (Transitions between states)
%     \draw[->] (NE) edge[loop left] node {$P(NE|NE, a)$} (NE);
%     \draw[->] (E) edge[loop above] node {$P(E|E, a)$} (E);
%     \draw[->] (SE) edge[loop right] node {$P(SE|SE, a)$} (SE);

%     % Transitions between NE and E
%     \draw[->] (NE) edge[bend left] node[above] {$P(E|NE, a)$} (E);
%     \draw[->] (E) edge[bend left] node[below] {$P(NE|E, a)$} (NE);

%     % Transitions between E and SE
%     \draw[->] (E) edge[bend left] node[above] {$P(SE|E, a)$} (SE);
%     \draw[->] (SE) edge[bend left] node[below] {$P(E|SE, a)$} (E);

% \end{tikzpicture}}
%     \caption{Caption}
%     \label{fig:enter-label}
% \end{figure}

\begin{figure}[h!]
    \centering
    % First minipage (2a)
        \resizebox{0.39\textwidth}{!}{\begin{tikzpicture}[->, >=stealth', shorten >=1pt, auto, thick, node distance=3cm]
            \tikzstyle{state} = [draw, circle, minimum size=2cm, text centered]

            % Nodes (States)
            \node[state, fill=red!50] (A) {Low};
            \node[state, fill=green!30] (B) [right of=A] {Medium};
            \node[state, fill=yellow!50] (C) [right of=B] {High};

            % Arrows (Transitions between states)
            \draw[->] (A) edge[loop left] node {$P(L|L, 0)$} (A);

            % Transitions between A and B

            \draw[->] (B) edge[bend left] node[below] {$P(L|M, 0)$} (A);

            % Transitions between B and C

            \draw[->] (C) edge[bend right] node[above] {$P(M|H, 0)$} (B);
        \end{tikzpicture}}

    \hfill
    % Second minipage (2b)

      \resizebox{0.48\textwidth}{!}{\begin{tikzpicture}[->, >=stealth', shorten >=1pt, auto, thick, node distance=3cm]
            \tikzstyle{state} = [draw, circle, minimum size=2cm, text centered]

            % Nodes (States)
            \node[state, fill=red!50] (A) {Low};
            \node[state, fill=green!30] (B) [right of=A] {Medium};
            \node[state, fill=yellow!50] (C) [right of=B] {High};

            % Arrows (Transitions between states)
            \draw[->] (A) edge[loop left] node {$P(L|L, 1)$} (A);

            \draw[->] (C) edge[loop right] node {$P(H|H, 1)$} (C);

            % Transitions between A and B
            \draw[->] (A) edge[bend left] node[above] {$P(M|L, 1)$} (B);
            \draw[->] (B) edge[bend left] node[below] {$P(L|M, 1)$} (A);

            % Transitions between B and C
            \draw[->] (B) edge[bend left] node[above] {$P(H|M, 1)$} (C);
            \draw[->] (C) edge[bend left] node[below] {$P(M|H, 1)$} (B);
        \end{tikzpicture}}

    \caption{General transition kernels with $a=0$ in above and $a=1$ in below for a patient in CPAP example.}
    \label{fig:example}
\end{figure}
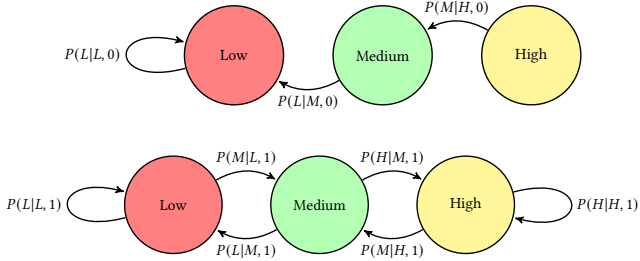

\begin{example}\label{example1}
\textbf{Continuous Positive Airway Pressure Therapy (CPAP).}
The CPAP \citep{herlihy2023planning,li2022towards, wang2024online} is a highly effective treatment when it is used
consistently during sleeping for adults with obstructive
sleep apnea. Since non-adherence to CPAP in patients hinders the effectiveness, we adapt the Markov model of CPAP
adherence behavior to a three-state system with the clinical adherence criteria. To elaborate, three distinct states are defined to characterize adherence levels: Low (L), Medium (M), and High (H) as shown in Figure \ref{fig:example}. Generally speaking, when action $a=0$ is taken, i.e., no intervention, the patient has a probability of $1$ to move from a higher adherence level to a lower adherence level. While intervention is available, a patient can either transit to a lower adherence lever or a higher adherence level with certain probabilities. In standard CPAP, the reward is set as the $1$ for state ``low adherence", $2$ for state ``medium adherence", and $3$ for state ``high adherence".
\end{example}

\begin{prop}\label{prop:indexability}
The MDP for each patient defined in Example \ref{example1} is indexable.
\end{prop}

Proposition \ref{prop:indexability} indicates that the Whittle index can be employed for the constructed CPAP problem in Example \ref{example1}. To verify that the Whittle index policy \cite{whittle1988restless} and LP-based policy \cite{xiong2021reinforcement,zhang2021restless} fail in this example, we construct the following setting. We randomly generate 20 different arms and each
arm is duplicated 10 times, whose transition probability matrices are generated randomly.  The budget is set to $K = 10$.  The objective is to maximize the total adherence level in a finite horizon $T=10$.
More importantly, each arm can only be pulled at most once.

Figure \ref{fig:failure example} highlights the performance limitations of existing policy strategies—specifically the LP-based method and the Whittle index policy—when \texttt{single-pull constraint} presents. It shows the normalized rewards where the optimal policy\footnote{Though we usually do not know the performance achieved by optimal policy, we can leverage the optimal value achieved for the LP in \eqref{eq:Obj_LP_final} to serve as the optimal performance, as it is always an upper bound of the optimal performance. This  will be explained in detail in Section \ref{Section4}.}, used as a benchmark, achieves a score of 1.0. The LP-based policy attains 0.76, and the Whittle index policy only 0.46 of the optimal performance. These results underline the ineffectiveness of both the LP-based method and the Whittle index policy in adequately handling the \texttt{single-pull constraint}, as neither approach reaches the efficiency of the optimal policy, particularly with the Whittle index-based approach performing less than half as well. This comparison suggests that these methods require modifications or alternative strategies to improve their adaptability and effectiveness under the strict limitations imposed by the \texttt{single-pull constraint}.

\begin{figure}[htbp]
\centering
\includegraphics[width=0.4\textwidth]{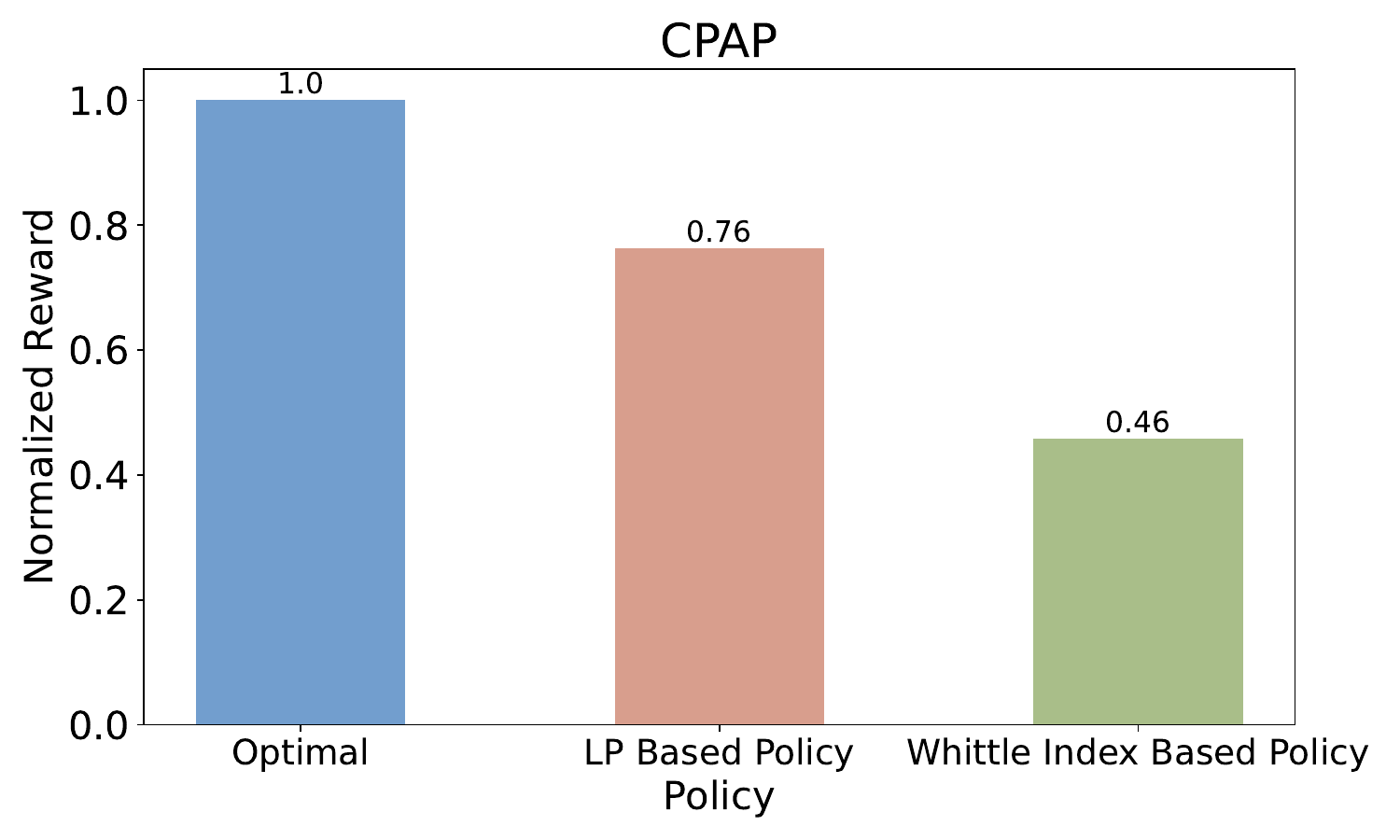}
\vspace{-0.1cm}
\caption{A CPAP setting with $3$ different states, $20$ different types of arms, each type has $10$ arms, the budget is set to be $10$ and the time Horizon is $10$.}
\label{fig:failure example}
\vspace{-0.2in}
\end{figure}
An intuitive explanation for why existing index policies fail in the single-pull setting is twofold. First, these indices are designed without accounting for the single-pull constraint. Second, traditional index policies pull arms from highest to lowest index until the activation budget is exhausted. However, in the \texttt{SPRMAB} setting, pulling the arm with the highest index at the current time may not lead to better results, as waiting for a future time slot could yield a higher reward. This makes the traditional approach ineffective in such scenarios.

\subsection{Challenge for Extending Existing Methods}
The limitations of existing index policies in addressing the \texttt{single-pull constraint} in \eqref{eq:one-pull-cons} become evident in the context of SPRMAB settings. Traditional policies such as the Whittle index fail to accommodate this constraint effectively because the additional dimensional constraint inherent in the single-pull scenario disrupts the foundational principles underpinning the Whittle index’s definition, rendering it inapplicable. Consequently, attention shifts to the LP-based index policy. This focus is due to the adaptability of LP approaches, which may allow for the integration of the \texttt{single-pull constraint} through modifications to the existing framework.  This approach requires re-evaluating the LP formulation to ensure it captures the critical aspects of decision-making under the stringent limitations imposed by the single-pull constraint.

Similar to  relaxing the ``hard'' \texttt{budget constraint}, we can also relax the \texttt{single-pull constraint} in \eqref{eq:one-pull-cons} so that the total number of pulls per arm is limited by 1 only in expectation as
\begin{align}\label{eq:relaxed_single_pull_cons}
\hspace{-0.2cm}\texttt{Re-single-pull constraint:}\mathbb{E}_{\pi}\left\{\sum_{t=1}^{T} a_n(t) \right\}\!\leq\! 1, \forall n.
\end{align}
Hence, we have the relaxed problem of \eqref{eq:One-Pull-Opt} expressed as

\begin{align}\label{eq:relaxed_One-Pull-Opt}\nonumber
\texttt{Re-SPRMAB}:&\max_{\pi}\mathbb{E}_\pi \left(\sum_{n=1}^{N} \sum_{t=1}^{T} r_n(t)\right)\\
&\mbox{s.t.}\mathbb{E}_{\pi}\left\{\sum_{n=1}^{N} a_n(t) \right\}\!\leq\! K, ~\mathbb{E}_{\pi}\left\{\!\sum_{t=1}^{T} a_n(t) \!\right\}\!\leq\! 1, \forall n.
\end{align}
According to the definition of OM $\mu$, the \texttt{Re-SPRMAB} in~\eqref{eq:relaxed_One-Pull-Opt} can be reformulated as the following LP \cite{altman2021constrained}:

\begin{align}\nonumber
\texttt{SPRMAB-LP:}~\max_{\mu}&\sum_{n=1}^{N}\sum_{t=1}^T\sum_{(s,a)} \mu_n(s,a;t){r}_n(s,a)\displaybreak[0]\\ \nonumber
\hspace{-0.1cm}\mbox{ s.t. }& Constriants~ \eqref{eq:con_rel_lp}-\eqref{eq:pmc_rel_lp2},\\   &\sum_{t=1}^T\sum_{s} \mu_n(s,1;t)\leq 1, \forall n. \label{eq:SPRMAB-LP}
\end{align}
It is clear that the \texttt{SPRAMB-LP} in \eqref{eq:SPRMAB-LP} achieves an upper bound of the optimal value  of \texttt{SPRAMB} in
\eqref{eq:One-Pull-Opt}, which is shown as the following proposition.

\begin{prop}\label{prop:upperbound1}
The optimal value achieved by \texttt{SPRMAB-LP} in \eqref{eq:SPRMAB-LP} is an upper bound of that of \texttt{SPRMAB} in \eqref{eq:One-Pull-Opt}.
\end{prop}

\begin{proof}[Proof Sketch]
 Since the \texttt{SPRMAB-LP}  in \eqref{eq:SPRMAB-LP} is equivalent to the relaxed problem \texttt{Re-SPRMAB} in \eqref{eq:relaxed_One-Pull-Opt} \cite{altman2021constrained},
it is sufficient to show that \texttt{Re-SPRMAB} in \eqref{eq:relaxed_One-Pull-Opt}  achieves no less average reward than the original problem \texttt{SPRMAB} in \eqref{eq:One-Pull-Opt}.
The proof is straightforward since the constraints in the relaxed problem  expand the feasible region of \texttt{SPRMAB} in \eqref{eq:One-Pull-Opt}. Denote the feasible region of \texttt{SPRMAB} in \eqref{eq:One-Pull-Opt}  as
\begin{align*}
    \Delta:=\Bigg\{a_n(t), \forall n, t\Bigg\vert\sum_{n=1}^{N}a_n(t)\leq K, \sum_{t=1}^T a_n(t)\leq 1\Bigg\},
\end{align*}
and the feasible region of \texttt{Re-SPRMAB} in \eqref{eq:relaxed_One-Pull-Opt} as
\begin{align*}
\Delta^\prime:=\left\{a_n(t), \forall n, t\Bigg\vert\mathbb{E}_{\pi}\left\{\sum_{n=1}^{N} a_n(t) \right\}\!\leq\! K, ~\mathbb{E}_{\pi}\left\{\!\sum_{t=1}^{T} a_n(t) \!\right\}\!\leq\! 1\right\}.
\end{align*}
It is clear that \texttt{Re-SPRMAB} expands the feasible region of \texttt{SPRMAB}, i.e., $\Delta\subseteq\Delta^\prime.$  Therefore, \texttt{Re-SPRMAB}  achieves an objective value no less than that of \texttt{SPRMAB} because the original optimal solution is also inside the relaxed feasibility set. This indicates \texttt{SPRMAB-LP} in \eqref{eq:SPRMAB-LP} achieves an optimal value no less than that of \eqref{eq:One-Pull-Opt}.
\end{proof}

One drawback of\texttt{SPRMAB-LP} in \eqref{eq:SPRMAB-LP} is that the mapping of \texttt{single-pull constraint}  from \eqref{eq:One-Pull-Opt} to the one in \eqref{eq:SPRMAB-LP} will make the ``hard'' constraint relaxed to a significant extent, as the probability of $\mu_n(s,1;t)$ will diffuse to different time steps, which contradicts with the real scenario where arms only be pulled for one particular time slot. This will make the associated index policy designed upon the solution of  \texttt{SPRMAB-LP} be significantly suboptimal. To make the relaxed problem tighter, we need to add the following constraint: for arbitrary time slot $t$, if arm $n$ is being activated, the arm should never be activated in other time slots $t^\prime$ with $t^\prime\neq t$, which can be mapped as
    \begin{align}\label{eq:add-on-constraint}
    {\sum_{s}} \mu_n(s,1;t)\cdot{\sum_{s}} \mu_n(s,1;t^\prime)=0.
\end{align}
Incorporating  the additional constraint in \eqref{eq:add-on-constraint} into \texttt{SPRMAB-LP} in \eqref{eq:SPRMAB-LP}, we have the following optimization problem:
\begin{align}\label{eq:one-pull-lp}\nonumber
\max_{\mu}&~~\sum_{n=1}^{N}\sum_{t=1}^T\sum_{(s,a)} \mu_n(s,a;t){r}_n(s,a)\displaybreak[0]\\ \nonumber
\hspace{-0.1cm}\mbox{ s.t. }& Constraints ~\eqref{eq:con_rel_lp}-\eqref{eq:pmc_rel_lp2},~ \eqref{eq:SPRMAB-LP}\\
&{\sum_{s}} \mu_n(s,1;t)\cdot{\sum_{s}} \mu_n(s,1;t^\prime)=0.
\end{align}
Regarding the novel optimization problem in \eqref{eq:one-pull-lp}, we have the following proposition.
\begin{prop}\label{prop:upper_bound2}

The optimal value achieved by \eqref{eq:one-pull-lp} lies in the middle of that by \texttt{SPRMAB-LP} in \eqref{eq:SPRMAB-LP} and that of \eqref{eq:One-Pull-Opt}.
\end{prop}

\begin{remark}\label{remark:1}
The proof of Proposition \ref{prop:upper_bound2} follows a similar argument as that in Proposition \ref{prop:upperbound1}. It indicates that the modified problem in \eqref{eq:one-pull-lp} achieves a tighter upper bound compared with the \texttt{SPRMAB-LP} in \eqref{eq:SPRMAB-LP}, and thus the associated index policy designed upon the solution of \eqref{eq:one-pull-lp} performs better than that derived from \eqref{eq:SPRMAB-LP}. However,  the above formulation in \eqref{eq:one-pull-lp} is not an LP any longer due to the last constraint, which leads to outstanding challenge in solving the above revised optimization problem when the state space $|\cS|$ and time horizon $T$ are large.
\end{remark}

% Given those constraints, we now have a new formulation as follows:
% \begin{align}\label{eq:one-pull-lp}\nonumber
% \max_{\mu}&~~\sum_{n=1}^{N}\sum_{t=1}^T\sum_{(s,a)} \mu_n(s,a;t)\bar{r}_n(s)\displaybreak[0]\\ \nonumber
% \hspace{-0.1cm}\mbox{ s.t. }& ~~{ \sum_{n=1}^{N}\sum_{s} \mu_n(s,1;t)\leq K}, ~\forall t,{\color{blue}//activation ~constraint }
% \displaybreak[1]\\ \nonumber
% &~~ {\sum_{a}} \mu_n(s,a;t)\!=\!\!\!\!\sum_{(s^\prime,a^\prime)}\!\!\!\mu_n(s^\prime, a^\prime; t\!-\!1)P_n(s^\prime, a^\prime,s),\forall n, s,{\color{blue}//fluid ~equilibrium }\displaybreak[3]\\ \nonumber
% &~~\sum_{a}\mu_n(s,a;1)={\bold{s}_1(s)},~\forall s, n,{\color{blue}//initial ~distribution }\\ \nonumber
% &\sum_{t=1}^T\sum_{s\in\cS}\mu_n(s,1;t)\leq 1, \forall n,{\color{blue}//maximal ~pull ~is ~1 }\\
% &{\sum_{s}} \mu_n(s,1;t)\cdot{\sum_{s}} \mu_n(s,1;t^\prime)=0 {\color{blue}//be ~pulled~ at ~one~ time~ slot }
% \end{align}

\section{Proposed Method}\label{Section4}
To address the challenge of solving the problem in \eqref{eq:one-pull-lp} posed by the nonlinear constraint in \eqref{eq:add-on-constraint} (as indicated in Remark \ref{remark:1}), we propose a novel method to handle the \texttt{single-pull constraint}. This method involves modifying the underlying Markov Decision Processes (MDPs) associated with the arms by introducing the concept of \emph{dummy states}. 

 In the considered \texttt{SPRMABs}, for arbitrary state of each arm $s\in\cS$ at current time step $t$, it transitions to next state $s^\prime\in\cS$ at time $t+1$ if a pull is assigned to this arm. After reaching state $s^\prime$ at time $t+1$, the arm will never be pulled again. For every pulled arm, regardless of its state or current time step, the available action set thereafter is restricted to $\{0\}$. 
Building on the aforementioned observation, we introduce \textit{dummy states} to represent the states reached immediately after an arm is pulled. We enforce that these dummy states have the same transition kernels and reward functions under both actions 0 and 1, identical to those of their corresponding normal states. The formal definition is given as follows.

\begin{defn}[Dummy state]
A dummy state $s_d$ represents the state $s\in\cS$ that being transited immediately when an arm is pulled. For a dummy state $s_d$ we have the following properties:
\begin{align}
    P_n(s_d, 0, s_d^\prime)&=P_n(s_d, 1, s_d^\prime)=P_n(s,0,s^\prime), \nonumber\\r_n(s_d,0)&=r_n(s_d,1)=r_n(s,0), \forall n,
\end{align}
i.e., actions $0$ and $1$ are indifferent in dummy states for all arms.
\end{defn}

\begin{remark}
For every state $s\in\cS$, it has a corresponding dummy state $s_d$. When introducing the dummy states, we duplicate the original state space $\cS$ and define the dummy state space as $\cS_d$. As a result, the system now has a new expanded state space $\cS^\prime:=\cS\bigcup\cS_d$. The intuitive idea behind introducing dummy states and enforcing indifference between actions 0 and 1 for these states is to ensure that the resource budget flows toward arms in non-dummy states, as arms in dummy states yield no gain even if resources are allocated to them. Another key advantage of using dummy states is that they allow us to eliminate the nonlinear constraint in \eqref{eq:add-on-constraint}. These points will be discussed in more detail in subsequent sections.
\end{remark}

To better understand how dummy states work, we present the following toy example.
\begin{example}
Consider a setting where the original state space is  $\mathcal{S} = \{s_0, s_1\} $. We introduce two corresponding dummy states,  $s_{0,d}$  and  $s_{1,d}$ , which are absorbing states. As a result, the expanded state space becomes  $\mathcal{S}^{\prime} := \{s_0, s_1, s_{0,d}, s_{1,d}\}$. In this setup, once an arm transitions into a dummy state (either  $s_{0,d}$  or  $s_{1,d}$ ), it remains in dummy states indefinitely, regardless of the action taken. Hence, both  $s_{0,d}$  and  $s_{1,d}$  are absorbing states, meaning that no matter which action is chosen (either action 0 or action 1), the transition probabilities from these dummy states remain the same. The transitions between these states are illustrated in Figure \ref{fig:4SMDP}.
Thus, arms in these dummy states provide no additional reward even for positive action assignment, ensuring that resources are directed toward arms in non-dummy states, which can still benefit from positive action assignments.
\end{example}

\begin{figure}
    \centering
    \begin{tikzpicture}[->, >=stealth', auto, semithick, node distance=2.7cm]
	\tikzstyle{every state}=[fill=white,draw=black,thick,text=black,scale=0.8]
	\node[state, fill=blue!30]    (A)                     {$s_0$};
	\node[state, fill=blue!30]    (B)[below of=A]   {$s_1$};
	\node[state, fill=green!30]    (C)[right of=A]   {$s_{0,d}$};
	\node[state, fill=green!30]    (D)[right of=B]   {$s_{1,d}$};
	\path
	(A) edge[loop left]			node{$P^0_{s_0,s_0}$}	(A)
        (A) edge node{$P^0_{s_0,s_1}$}   (B)
	(B) edge[bend left]	node{$P^0_{s_1,s_0}$}	(A)
        (B) edge[loop left]			node{$P^0_{s_1,s_1}$}	(B)
        (C) edge[loop right]			node{$P^{1}_{s_{0,d},s_{0,d}}=P^{0}_{s_{0,d}, s_{0,d}}$}	(C)
        (C) edge[bend left]  node{$P^{1}_{s_{0,d},s_{1,d}}=P^{0}_{s_{0,d},s_{1,d}}$}   (D)
	(D) edge	node{$P^{1}_{1_d,0_d}$}	(C)
        (D) edge[loop right]			node{$P^{1}_{s_{1,d},s_{1,d}}=P^{0}_{s_{1,d},s_{1,d}}$}	(D)
	(A) edge[bend left] node{$P^{1}_{s_0,s_{0,d}}$}  (C)
        (A) edge[pos=0.25,above] node{$P^1_{s_0,s_{1,d}}$} (D)
        (B) edge[pos=0.30,below] node{$P^1_{s_1,s_{0,d}}$} (C)
        (B) edge[bend right] node{$P^1_{s_1,s_{1,d}}$} (D);

	%\node[above=0.5cm] (A){Patch G};
	%\draw[red] ($(D)+(-1.5,0)$) ellipse (2cm and 3.5cm)node[yshift=3cm]{Patch H};
	\end{tikzpicture}
    \caption{A toy example of \texttt{SPRAMB} with dummy states. The original state space is $\cS=\{s_0,s_1\}$, and it leads to a 4-state expanded system as $\cS^\prime=\{s_0, s_1, s_{0,d}, s_{1,d}\}$. }
    \label{fig:4SMDP}
    \vspace{-0.2in}
\end{figure}
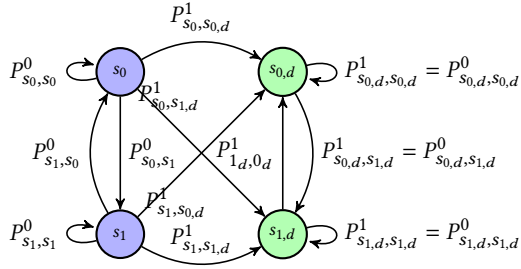

Once we provide the new state space $\cS^\prime$ containing dummy states and the new transition kernels, we have a new formulation as follows:
\begin{align}\nonumber
\max_{\mu}&~~\sum_{n=1}^{N}\sum_{t=1}^T\sum_{s\in\cS^\prime}\sum_{a\in\cA} \mu_n(s,a;t){r}_n(s,a)\displaybreak[0]\\ \nonumber
\hspace{-0.1cm}\mbox{ s.t. }& ~~{ \sum_{n=1}^{N}\sum_{s\in\cS^\prime} \mu_n(s,1;t)\leq K}, ~\forall t,\displaybreak[1]\\ \nonumber
& {\sum_{a\in\cA}} \mu_n(s,a;t)=\!\!\sum_{s^\prime\in\cS^\prime}\!\!\sum_{a^\prime\in\cA}\!\!\mu_n(s^\prime, a^\prime; t-1)P_n(s^\prime, a^\prime,s), \forall n,s\in\cS^\prime,\displaybreak[3]\\
&\sum_{a}\mu_n(s,a;1)={\bold{s}_1(s)},
~\forall s\in\cS, n.
\label{eq:Obj_LP_final}    
\end{align}
\begin{remark}
   Given the expanded state space $\cS^\prime$
  with dummy states and the modified transition kernels, we can remove both the \texttt{single-pull constraint} in \eqref{eq:SPRMAB-LP}
 and the nonlinear constraint in \eqref{eq:add-on-constraint}. The reason we can eliminate the \texttt{single-pull constraint} is that once an arm is pulled, it transitions to a dummy state, where there is no difference between action 1 and action 0. As a result, when arms in non-dummy states can benefit from a positive action, the active OM naturally flows to those arms. Similarly, the nonlinear constraint to concentrate the OM in \eqref{eq:add-on-constraint}
 is no longer required, as it becomes irrelevant without the \texttt{single-pull constraint}.

\end{remark}

\subsection{The Single-Pull Index Policy}

Once we solve \eqref{eq:Obj_LP_final} and get the optimal $\mu^\star$, we can define the following Markovian policy\footnote{If the denominator equals to 0, we direct set $\chi\ust_n(s,1;t)=0$.}
\begin{align}
\chi\ust_n(s,1;t):=
\frac{\mu\ust_n(s,1;t)}{\mu\ust_n(s,0; t)+\mu\ust_n(s,1; t)}\in[0,1],
\label{def:policy}
\end{align}
which denotes the probability of selecting action $1$ for arm $n$ with state $s$ at time $t$. 
  Note that the optimal policy~(\ref{def:policy}) is not always feasible for \texttt{SPRMAB} since in the latter at most $K$ units of activation costs can be consumed at a time and the arm can only be pulled once. To this end, we construct our single-pull index (\texttt{SPI})  $\cI_n(s_n(t);t)$ associated with arm $n$ at time $t$ as 
\begin{align}\label{eq:occupancy-index}
    \cI_n(s_n(t);t):=\chi\ust_n(s_n(t),1;t)r_n(s_n(t),1), 
\end{align}
where $\chi\ust_n(s_n(t),a;t)$ is defined in \eqref{def:policy}.
Notice that our \texttt{SPI} measures the expected reward for activating the arm $n$ in state $s_n(t)$ at time $t$. Therefore, a higher index indicates a higher expected reward, suggesting that the intuition is to activate arms that contribute more significantly to the accumulated reward.
Our index policy then activates arms with \texttt{SPI} indices in a decreasing order.  The entire procedure is summarized in Algorithm \ref{alg:Index}.

\begin{algorithm}
	\caption{\texttt{SPI} Index Policy}
	\label{alg:Index}
\textbf{Input}: 
	Initialize $\bs_1(s)$ $\forall n\in[N].$
 \begin{algorithmic}[1]
	 	\STATE Construct the LP according to \eqref{eq:Obj_LP_final} 
	  and solve the occupancy measure $\mu^*$; 
			\STATE Compute $\chi_n^\star(s,a,t),  \forall s, a, t$ according to  \eqref{def:policy};
		\STATE Construct the \texttt{SPI} set $\cI(t):=\{I_n(s_n(t);t): n\in[N]\}$ according to \eqref{eq:occupancy-index}; and sort $\cI(t)$ in a decreasing order; \label{step:algo}
		\IF{Budget remains}
	\STATE Activate arms according to the order in step \ref{step:algo} ;
 \IF{the activated arm is in dummy states}
\STATE Do not pull the arm and let the budget minus one unit;
 \ENDIF
\ENDIF	
\end{algorithmic}
\end{algorithm}

% \begin{remark}
%    The introduction of dummy states and the proposed Algorithm \ref{alg:Index} offers two significant advantages when dealing with challenges in \texttt{SPRMABs}. The first challenge is that, in real-world applications, pulling an arm is not always guaranteed to be better than not pulling it. As a result, the total budget 
% $K$ may not need to be fully utilized, whereas existing index policies \cite{whittle1988restless, xiong2021reinforcement,zhang2021restless,hu2017asymptotically, verloop2016asymptotically}, tend to imprudently pull 
% $K$ arms with the highest indices, leading to a sub-optimal policies.
% The second challenge is specific to \texttt{SPRMABs}, where each arm can only be pulled once. The time slot in which an arm is pulled becomes critical. Even if an arm has the highest index at the current state and time slot, it may not be the optimal time to activate it. Waiting for a future time slot to pull the arm could lead to a greater reward. This issue is not handled well by existing algorithms.
% By introducing dummy states, the probability of pulling an arm can effectively flow into the dummy state, since pulling or not pulling an arm in a dummy state has no impact on future rewards. This allows us to conserve the budget for dummy states, meaning we do not need to exhaust the entire budget, addressing the two challenges prudently.
% \end{remark}

\begin{remark}
    The introduction of dummy states and Algorithm 1 addresses two key challenges in SPRMABs. First, in real-world applications, pulling an arm isn't always better than not pulling it, meaning the total budget 
$K$ may not need to be fully used. Existing index policies often pull 
$K$ arms with the highest indices, leading to suboptimal decisions. Second, in SPRMABs, where each arm can only be pulled once, the timing of pulling an arm is critical. Even if an arm has the highest index now, waiting for a future time slot may yield a higher reward, a factor overlooked by existing algorithms. Dummy states allow for conserving budget, as pulling an arm in a dummy state has no impact on future rewards, addressing both challenges effectively.
\end{remark}

\subsection{Asymptotic and Non-Asymptotic Optimality}\label{sec:thm_result}
We now provide results on asymptotic and non-asymptotic optimality for our new index. We begin by showing that our index is asymptotically optimal in the same asymptotic regime as that in Whittle \cite{whittle1988restless} and others \cite{verloop2016asymptotically,weber1990index,zhang2021restless,xiong2021reinforcement}.  With some abuse of notation, let the number of users be $\rho N$ and the resource constraint be $\rho K$ in the asymptotic regime with $\rho\rightarrow\infty.$  In other words, we consider $N$ different classes of users with each class containing $\rho$ users. Let $J^{\pi}(\rho K, \rho N)$ denote the expected total reward of the original problem~(\ref{eq:one-pull-cons}) under an arbitrary policy $\pi$ for such a system. Denote the optimal policy for the original problem~(\ref{eq:One-Pull-Opt}) as $\pi^{Opt}:=\{\pi_n^{Opt}, \forall n\in\mathcal{N}\}$.

\begin{thm}\label{thm:asym_opt}
The designed \texttt{SPI} policy (Algorithm~\ref{alg:Index}) is asymptotically optimal, i.e., 
\begin{align}\label{eq:Asym_opt}
\lim_{\rho\rightarrow \infty} \frac{1}{\rho N}\Big(J^{\pi^{Opt}}(\rho K, \rho N)- J^{\pi^{SPI}}(\rho K, \rho N)\Big)=0.
\end{align}
\end{thm}

\begin{remark}
Theorem \ref{thm:asym_opt} indicates that as the number of per-class users goes to infinity, the average gap between the performance achieved by our \texttt{SPI} policy $ \pi^{SPI}$ and the optimal policy $\pi^{Opt}$ tends to be zero. This is a well-established criteria for showing the "optimality" of designed index policies in existing work \cite{xiong2021reinforcement,whittle1988restless, zhang2021restless, hu2017asymptotically, verloop2016asymptotically}, and the results only hold for $\rho\rightarrow\infty$. In the following, we present a more rigorous characterization of the optimality gap under finite scaling factor $\rho$, which is given by Theorem  \ref{thm:non-asymptotic}. 
\end{remark}

\begin{thm}\label{thm:non-asymptotic}
For a finite number of $\rho\in\mathbb{R}^+$, with probability at least $1-\frac{1}{\rho}$ such that the average gap between the performance achieved by our index policy $ \pi^{Index}$ and the optimal policy $\pi^{SPI}$ is given as 
\begin{align}
    &\frac{1}{\rho N}\Big(J^{\pi^{Opt}}(\rho B, \rho N)- J^{\pi^{SPI}}(\rho B, \rho N)\Big)\nonumber\\
    &\leq 2r_{\text{max}}T\sqrt{\frac{ \ln 2\rho}{2N\rho}}+5r_{\text{max}}T\sqrt{\frac{\ln 2\rho}{2\rho^3}}. 
\end{align}
\end{thm}
\begin{cor}
    Theorem \ref{thm:non-asymptotic} indicates that the average gap between the
performance achieved by our index policy $\pi^{Index}$ and the optimal
policy $\pi^{Opt}$ is of the order of $\mathcal{O}\left(\frac{1}{\rho^{1/2}}+\frac{1}{\rho^{3/2}}\right)$, which is dominated by the first term. We also observe from Theorem \ref{thm:non-asymptotic} that if multiple factor $\rho$ goes to infinity, the performance gap converges to $0$, which resume the asymptotic optimality in Theorem \ref{thm:asym_opt}.
\end{cor}

\begin{remark}
Our index policy is computationally appealing since it is only based on the ``relaxed problem'' by solving a LP.
Furthermore, if all arms share the same MDP, the LP can be decomposed across arms as in \cite{whittle1988restless}, and hence the computational complexity does not scale with the number of arms.  More importantly, our index policy is well-defined without the requirement of indexability condition \cite{whittle1988restless}.  This is in contrast to most of the existing Whittle index-based policies that are only well defined in the case that the system is indexable, which is hard to verify and may not hold in general. Closest to our work is the parallel work on restless bandits \cite{zhang2021restless}, which explores index policies similar to ours, but under the assumption of homogeneous MDPs across arms in the binary action settings, and mainly focus on characterizing the asymptotic optimality gap. There is a branch of work \cite{zhang2021restless,gast2023linear,brown2023fluid, ghosh2022indexability} that focuses on analyzing the gap with an explicit relationship involving 
$\rho$, but all of them rely on the assumption that 
$rho$ is large enough for the central limit theorem and mean-field approximation to apply. In contrast, we provide the first characterization that holds for a finite number of 
$\rho$. This further differentiates our work with existing literature.

\end{remark}

\begin{table*}[htbp]
    \centering
    \scalebox{0.9}{
    \begin{tabularx}{\textwidth}{l>{\centering\arraybackslash}X>{\centering\arraybackslash}X>{\centering\arraybackslash}X>{\centering\arraybackslash}X>{\centering\arraybackslash}X>{\centering\arraybackslash}X>{\centering\arraybackslash}X>{\centering\arraybackslash}X>{\centering\arraybackslash}X}
    \toprule
    \multirow{2}{*}{\textbf{Policy}} & \multicolumn{3}{c}{\textbf{Birth-Death Process (CPAP)}} & \multicolumn{3}{c}{\textbf{Greedy-Reliable-Fixed (MHMH)}} & \multicolumn{3}{c}{\textbf{Greedy-Reliable-Variable (MHMH)}} \\
    \cmidrule(lr){2-4} \cmidrule(lr){5-7} \cmidrule(lr){8-10}
    & \scriptsize{\textbf{$(20, 5, 10, 10, 10)$}} & \scriptsize{\textbf{$(40, 5, 10, 10, 10)$}} & \scriptsize{\textbf{$(40, 5, 10, 5, 12)$}} & \scriptsize{\textbf{$(10, 3, 25, 50, 10)$}} & \scriptsize{\textbf{$(20, 3, 50, 50, 10)$}} & \scriptsize{\textbf{$(20, 3, 25, 50, 20)$}} & \scriptsize{\textbf{$(20, 3, 1, 2, 20)$}}  & \scriptsize{\textbf{$(20, 3, 5, 10, 20)$}} & \scriptsize{\textbf{$(20, 3, 15, 30, 20)$}} \\
    \midrule
    Upper Bound& \colorbox{white}{200.0} & \colorbox{white}{200.0} & \colorbox{white}{220.0} & \colorbox{white}{174.6} & \colorbox{white}{343.6} & \colorbox{white}{415.9} & \colorbox{white}{16.6} & \colorbox{white}{83.2} & \colorbox{white}{249.5} \\
    SPI & \colorbox{green}{$197.5\pm 0.3$} & \colorbox{green}{$200\pm 0$} & \colorbox{green}{$217.4\pm 0.3$} & \colorbox{green}{$172.6\pm 0.3$} & \colorbox{green}{$341.9\pm 0.5$} & \colorbox{green}{$410.6\pm 0.5$} & \colorbox{yellow}{$14.0\pm 0.2$} & \colorbox{green}{$79.1\pm 0.3$} & \colorbox{green}{$244.5\pm 0.4$} \\
    Mean Field & \colorbox{white}{$187.8\pm 0.7$} & \colorbox{white}{$182.0\pm 0.3$} & \colorbox{yellow}{$211.4\pm 0.3$} & \colorbox{white}{$156.9\pm 0.4$} & \colorbox{white}{$311.6\pm 0.6$} & \colorbox{white}{$386.0\pm 0.7$} & \colorbox{white}{$13.6\pm 0.2$} & \colorbox{white}{$75.3\pm 0.3$} & \colorbox{white}{$230.3\pm 0.5$} \\
    Finite Whittle & \colorbox{white}{$153.6\pm 0.4$} & \colorbox{white}{$158.0\pm 0.1$} & \colorbox{white}{$153.7\pm 0.4$} & \colorbox{white}{$166.0\pm 0.3$} & \colorbox{white}{$322.1\pm 0.5$} & \colorbox{yellow}{$402.2\pm 0.5$} & \colorbox{green}{$14.6\pm 0.1$} & \colorbox{green}{$79.2\pm 0.2$} & \colorbox{yellow}{$240.8\pm 0.4$} \\
    Infinite Whittle & \colorbox{green}{$197.7\pm 0.2$} & \colorbox{green}{$200\pm 0$} & \colorbox{green}{$217.7\pm 0.3$} & \colorbox{green}{$172.1\pm 0.3$} & \colorbox{white}{$325.5\pm 0.5$} & \colorbox{white}{$379.8\pm 0.5$} & \colorbox{yellow}{$14.0\pm 0.1$} & \colorbox{white}{$75.0\pm 0.3$} & \colorbox{white}{$227.5\pm 0.4$} \\
    Original Whittle & \colorbox{white}{$178.2\pm 0.2$} & \colorbox{white}{$176.7\pm 0.2$} & \colorbox{white}{$199.3\pm 0.3$} & \colorbox{white}{$140.5\pm 0.5$} & \colorbox{white}{$250.3\pm 0.6$} & \colorbox{white}{$339.4\pm 0.6$} & \colorbox{white}{$12.2\pm 0.2$} & \colorbox{white}{$66.8\pm 0.3$} & \colorbox{white}{$202.8\pm 0.5$} \\
    Q-Difference & \colorbox{white}{$111.6\pm 0.4$} & \colorbox{white}{$112.2\pm 0.3$} & \colorbox{white}{$112.0\pm 0.3$} & \colorbox{white}{$165.7\pm 0.4$} & \colorbox{yellow}{$337.1\pm 0.5$} & \colorbox{white}{$388.5\pm 0.5$} & \colorbox{yellow}{$14.0\pm 0.1$} & \colorbox{yellow}{$76.3\pm 0.3$} & \colorbox{white}{$232.3\pm 0.4$} \\
    Random & \colorbox{white}{$140.0\pm 0.5$} & \colorbox{white}{$139.8\pm 0.4$} & \colorbox{white}{$159.8\pm 0.4$} & \colorbox{white}{$26.3\pm 0.3$} & \colorbox{white}{$47.2\pm 0.4$} & \colorbox{white}{$45.2\pm 0.5$} & \colorbox{white}{$1.8\pm 0.1$} & \colorbox{white}{$8.9\pm 0.3$} & \colorbox{white}{$26.9\pm 0.4$} \\
    \bottomrule
    \end{tabularx}
    }
    \caption{We present the performance of various policies across different domains and settings.
    We run each settings for $1000$ simulations and present $95\%$ confidence interval.
    Each setting is denoted by the parameters (number of types $N$, number of states $S$, budget $K$, group size $\rho$, time horizon $T$).
    In each simulation, we consider $N$ types of arms, with each type consisting of $\rho$ arms. Transition probabilities for each type are randomly assigned in every setting. Optimal policies are highlighted in green, and near-optimal policies are highlighted in yellow. Here near-optimal means the gap between it and optimal policy is less than three percent of the upper bound. We use the optimal value achieved for the LP \eqref{eq:Obj_LP_final} in \texttt{SPI} to serve as the upper bound.}
    \label{tab:algo_comparison}
    \vspace{-0.2in}
\end{table*}

\section{Experiments}\label{Sec:exp}
In this section, we numerically evaluate the proposed \texttt{SPI} policy in three domains: two from real-world applications and one from a synthetic domain, comparing it to state-of-the-art benchmark algorithms. Main results from the two real-world domains are presented here, while the results from the synthetic domain are provided in the supplementary materials.

\subsection{Benchmarks}
The benchmarks we compare in this paper are listed below:

\indent $\rhd$ \emph{Mean-Field LP-based index policy \cite{ghosh2022indexability}:} The mean-field LP-based index policy is a classic LP-based approach for solving RMAB problems, leveraging mean-field approximation theory when the number of arms is large, as expressed in \eqref{eq:obj_rel_lp}-\eqref{eq:pmc_rel_lp2}. However, it does not account for the single-pull constraint when designing the indices.

$\rhd$ \emph{Original Whittle index policy \cite{whittle1988restless}:} The Whittle index defined in \eqref{eq:Whittle_index} is the most widely used approach for solving Restless  RMABs. It is designed for infinite-horizon problems and does not take the single-pull constraint into account.

$\rhd$ \emph{Q-Difference policy \cite{biswas2021learn}:} The Q-difference method designs indices based on the difference between Q-value functions. It is a heuristic approach that can perform well in practice, but it lacks theoretical guarantees, making its performance uncertain in certain scenarios.

$\rhd$ \emph{Modified infinite Whittle index policy:} This is a hurestic modification for original Whittle index by considering the dummy states introduced for our proposed \texttt{SPI} policy in Section \ref{Section4}.

$\rhd$ \emph{Modified Finite Whittle index policy:} This is a further modification of modified Whittle index by considering finite-horizon time-dependent index.

% $\rhd$ \emph{Random policy:} The random policy is a basic benchmark used for comparison, where arms are selected to be pulled randomly. It serves as a lower bound for performance, offering a baseline against which other, more sophisticated policies can be evaluated.

\subsection{Experimental Domains} We briefly introduce the two considered real-world domains below, and relegate the detailed description to  Section \ref{sec:C1} in supplementary materials.
\subsubsection{Continuous Positive Airway Pressure Therapy (CPAP)\citep{herlihy2023planning,li2022towards, wang2024online}} CPAP
is a highly effective treatment for adults with obstructive sleep apnea when used consistently during sleep. We model CPAP adherence behavior as a multi-state system, adapting the Markov model with clinical adherence criteria, which reduces to a standard Birth-Death process. In the standard CPAP setting, lower adherence levels yield lower rewards.
The objective is to maximize the accumulated reward over time, with the constraint that each patient (arm) can only be pulled (intervened) at most once.

\subsubsection{Mobile Healthcare for Maternal Health (MHMH)\cite{ghosh2022indexability}}
In this program, healthcare workers make phone calls to enrollees (beneficiaries) to enhance engagement and provide targeted health information. Since the number of healthcare workers is much smaller than the number of beneficiaries, they must continuously prioritize which beneficiaries to call to maximize the total return. As in \cite{ghosh2022indexability}, we assume two types of beneficiaries, greedy and reliable. We intentionally show two variations of this domain ``fixed'' and ``variable'' to illustrate the impact of fixed and varying group sizes on the performance of different index policies.

% \subsubsection{Enhrenfest Project} We consider another domain called Ehrenfest project RMAB, whose whittle index was derived in \cite{whittle1988restless}.  More details and results are presented in Section \ref{sec:C1}.

% Our motivation for Ehrenfest project RMAB comes from original use to derive Whittle indices \cite{whittle1988restless}. More details and results in Section \ref{sec:C1}.

\vspace{-0.1in}
\subsection{Numerical Results}

The simulations are conducted for  $N$  types of arms, with each type consisting of  $\rho$  arms. Each type of arm follows a distinct MDP, and different types of arms have varying MDPs. All results are averaged over $1000$ Monte Carlo simulations to ensure robust performance evaluation. Due to the space limitation, we present the main results in this section and more numerical results can be found in Section \ref{sec:C2} in supplementary materials.

%\textit{SPI is seen to be a top-performer, matches only by infinite-whittle; however, infinite-whittle has significantly longer run-time. }
%{\color{blue}[discussion]}\yuqi{feel like we don't need to add this sentence since discussions below already include it}

\textbf{Accumulated Reward.} We first compare the accumulated reward performance for the proposed \texttt{SPI} index policy with all benchmark polices. Based on the numerical evaluation results presented in Table \ref{tab:algo_comparison}, we observe the performance of various policies across different domains and settings, with the proposed \texttt{SPI} policy consistently achieving near-optimal performance. Each setting is described by parameters such as the number of types
$N$, number of states
$S$, budget
$K$, group size
$\rho$, and time horizon
$T$. In all scenarios, \texttt{SPI} policy either matches or comes extremely close to the optimal policy, with a performance gap of less than 3\%,  demonstrating the robustness of \texttt{SPI} policy across varying settings. Other benchmark policies, such as the Mean Field policy and infinite Whittle index policy, often perform well but exhibit noticeable gaps from the optimal in certain cases. For instance, in the "Greedy-Reliable-Fixed" setting
$(20,3,50,50,20)$, Mean Field index policy achieves
$386.0\pm 0.7$, significantly lower than the optimal
$415.9$, indicating sub-optimality. The Random policy consistently underperforms, yielding the lowest rewards across all settings. For example, in the "Greedy-Reliable-Variable" setting
$(20,3,15,30,20)$, Random  policy achieves
$26.9\pm 0.4$, a stark contrast to the optimal
$249.5$. Both the finite and infinite Whittle index policies show decent performance, but often fall short compared to \texttt{SPI} policy and the optimal policy. In the "Birth-Death Process" setting
$(40,5,10,10,10)$, infinite Whittle index policy achieves the optimal
$200\pm 0$, while finite Whittle index policy lags behind with
$158.0\pm 0.1$. Overall, \texttt{SPI} policy consistently demonstrates strong performance, closely matching or achieving optimal rewards across all settings, while other benchmark methods show varying degrees of sub-optimality, with sometimes substantial gaps from the optimal policy.

We also observe that the performance of \texttt{SPI} improves and becomes better as the group size increases. In the ``Greedy-Reliable-Variable'' setting, when the group size is small, such as $\rho = 2$,  policies—including \texttt{SPI}, finite Whittle index, infinite Whittle index, and Q-difference—are either optimal or near-optimal. However, as the group size increases to $\rho = 30$, only \texttt{SPI} policy remains optimal, while finite Whittle index policy becomes sub-optimal and the performance of other policies declines even further.

We consider another domain called \emph{Enhrenfest project} studied in \cite{whittle1988restless}. For domains like the Enhrenfest project, different index
policies can achieve very similar near-optimal performance in practice, and the detail is presented in Section \ref{sec:C1}.

% \subsubsection{Enhrenfest Project} We consider another domain called Ehrenfest project RMAB, whose whittle index was derived in \cite{whittle1988restless}.  More details and results are presented in Section \ref{sec:C1}.

\begin{figure}[h]
%\vspace{-0.2in}
\centering
\includegraphics[width=0.4\textwidth, height=1.7in]{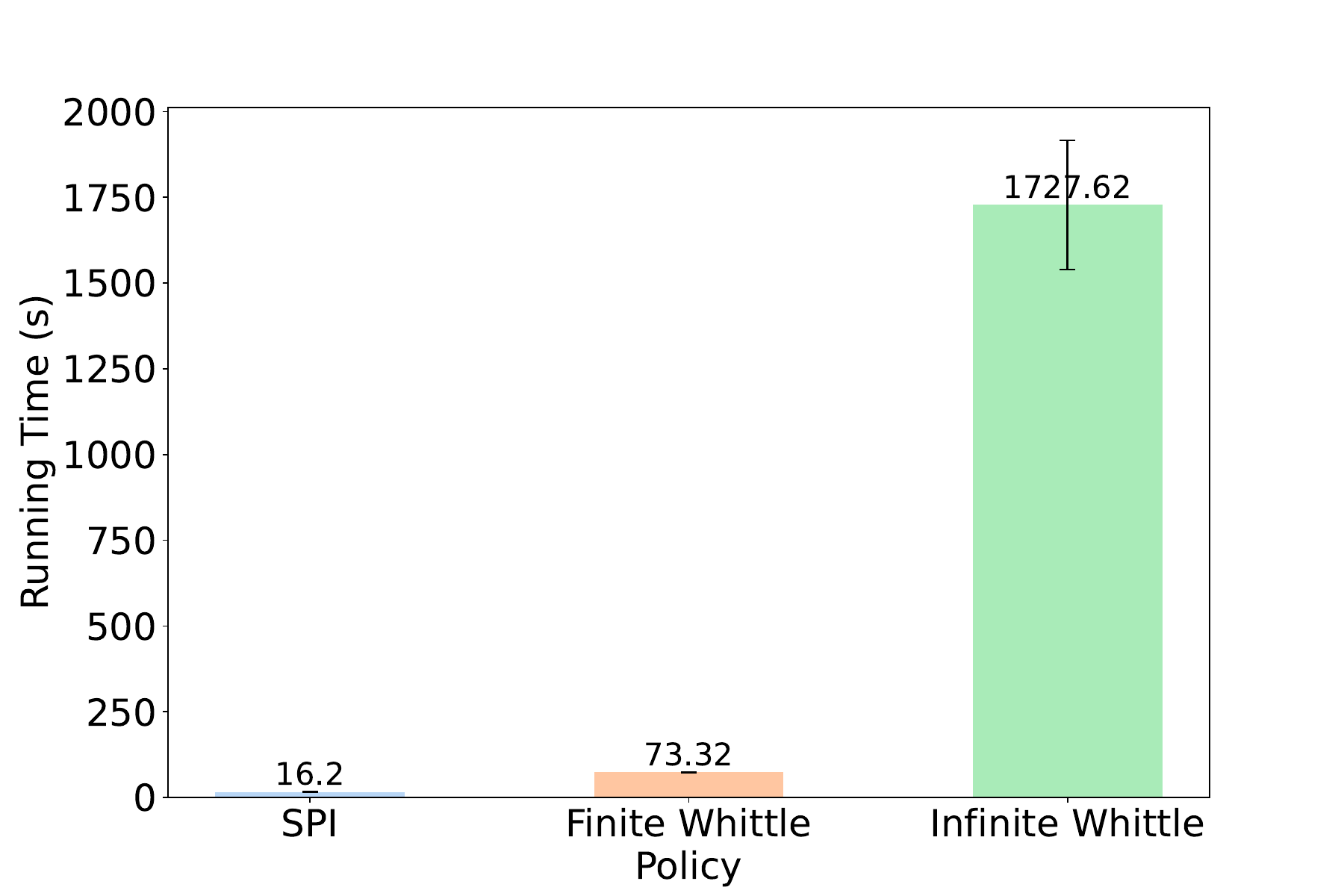}
\vspace{-0.1in}
\caption{We present the average running time of \texttt{SPI} policy, finite whittle policy, and infinite whittle policy in the CPAP setting $(N,S,K,\rho,T)=(10, 10, 50, 50, 10)$.}
\label{fig:BD_run_time}
%\vspace{-0.3in}
\end{figure}

\textbf{Running Time.} In Figure \ref{fig:BD_run_time}, we compare the running time of \texttt{SPI} policy with whittle-index-based policies for Birth-Death Process, and we include the running time comparison for randomly generated MDPs as a robustness check in Figure \ref{fig:Random_run_time} (Section \ref{sec:C2} in supplementary materials). We randomly generate three different  Birth-Death Process MDPs, and we take the average running time of each policy under different MDPs. The proposed \texttt{SPI} policy significantly outperforms the Whittle-index-based policy in terms of running time. The Whittle-index-based policy requires first computing the value function and then gradually adjusting the parameter to search the Whittle index, which results in a longer running time due to the curse of dimensionality. In contrast, the \texttt{SPI} policy only needs to solve the LP once when computing the index, which greatly enhances its scalability. This makes \texttt{SPI} particularly well-suited for real-time resource allocation, where non-profit organizations often face tight computation constraints. \textit{Although the infinite Whittle-index policy achieves competitive performance in some domains, its running time is significantly higher than that of the \texttt{SPI} policy, making \texttt{SPI} a more practical and efficient choice in scenarios where rapid decision-making is essential. }

\textbf{Asymptotic Optimality.} We empirically show that \texttt{SPI} is asymptotic optimal, and defer its discussion to Section \ref{subsubsec:asymptotic_optimality}.

\section{Conclusions}

Motivated by many real world resource allocation challenges, this paper introduces Finite-Horizon Single-Pull RMABs (SPRMABs), a novel variant of RMAB in which each arm can only be pulled once. This single-pull constraint introduces additional complexity, leading to ineffectiveness of previously proposed index policies. To address this limitation, we design a  lightweight index policy for this expanded system.  For the first time, we demonstrate that our index policy achieves a sub-linearly decaying average optimality gap.  Extensive simulations validate the proposed method, showing robust performance across various domains compared to existing benchmarks.

%%%%%%%%%%%%%%%%%%%%%%%%%%%%%%%%%%%%%%%%%%%%%%%%%%%%%%%%%%%%%%%%%%%%%%%%

%%% The acknowledgments section is defined using the "acks" environment
%%% (rather than an unnumbered section). The use of this environment
%%% ensures the proper identification of the section in the article
%%% metadata as well as the consistent spelling of the heading.

% \begin{acks}
% If you wish to include any acknowledgments in your paper (e.g., to
% people or funding agencies), please do so using the `\texttt{acks}'
% environment. Note that the text of your acknowledgments will be omitted
% if you compile your document with the `\texttt{anonymous}' option.
% \end{acks}

%%%%%%%%%%%%%%%%%%%%%%%%%%%%%%%%%%%%%%%%%%%%%%%%%%%%%%%%%%%%%%%%%%%%%%%%

%%% The next two lines define, first, the bibliography style to be
%%% applied, and, second, the bibliography file to be used.

\bibliographystyle{ACM-Reference-Format}
\bibliography{sample}


\begin{thebibliography}{51}

%%% the Web address.

\ifx \showCODEN    \undefined \def \showCODEN     #1{\unskip}     \fi
\ifx \showDOI      \undefined \def \showDOI       #1{#1}\fi
\ifx \showISBNx    \undefined \def \showISBNx     #1{\unskip}     \fi
\ifx \showISBNxiii \undefined \def \showISBNxiii  #1{\unskip}     \fi
\ifx \showISSN     \undefined \def \showISSN      #1{\unskip}     \fi
\ifx \showLCCN     \undefined \def \showLCCN      #1{\unskip}     \fi
\ifx \shownote     \undefined \def \shownote      #1{#1}          \fi
\ifx \showarticletitle \undefined \def \showarticletitle #1{#1}   \fi
\ifx \showURL      \undefined \def \showURL       {\relax}        \fi
% invisible to TeX
\providecommand\bibfield[2]{#2}
\providecommand\bibinfo[2]{#2}
\providecommand\natexlab[1]{#1}
\providecommand\showeprint[2][]{arXiv:#2}

\bibitem[\protect\citeauthoryear{Altman}{Altman}{2021}]%
        {altman2021constrained}
\bibfield{author}{\bibinfo{person}{Eitan Altman}.}
  \bibinfo{year}{2021}\natexlab{}.
\newblock \bibinfo{booktitle}{\emph{Constrained Markov decision processes}}.
\newblock \bibinfo{publisher}{Routledge}.
\newblock


\bibitem[\protect\citeauthoryear{Avrachenkov, Borkar, and
  Pattathil}{Avrachenkov et~al\mbox{.}}{2017}]%
        {avrachenkov2017controlling}
\bibfield{author}{\bibinfo{person}{K Avrachenkov}, \bibinfo{person}{VS Borkar},
  {and} \bibinfo{person}{S Pattathil}.} \bibinfo{year}{2017}\natexlab{}.
\newblock \showarticletitle{Controlling G-AIMD using index policy}. In
  \bibinfo{booktitle}{\emph{The 56th IEEE Conference on Decision and Control,
  Melbourne, December}}. \bibinfo{pages}{12--15}.
\newblock


\bibitem[\protect\citeauthoryear{Avrachenkov and Borkar}{Avrachenkov and
  Borkar}{2022}]%
        {avrachenkov2022whittle}
\bibfield{author}{\bibinfo{person}{Konstantin~E Avrachenkov} {and}
  \bibinfo{person}{Vivek~S Borkar}.} \bibinfo{year}{2022}\natexlab{}.
\newblock \showarticletitle{Whittle index based Q-learning for restless bandits
  with average reward}.
\newblock \bibinfo{journal}{\emph{Automatica}}  \bibinfo{volume}{139}
  (\bibinfo{year}{2022}), \bibinfo{pages}{110186}.
\newblock


\bibitem[\protect\citeauthoryear{Ayer, Zhang, Bonifonte, Spaulding, and
  Chhatwal}{Ayer et~al\mbox{.}}{2019}]%
        {ayer2019prioritizing}
\bibfield{author}{\bibinfo{person}{Turgay Ayer}, \bibinfo{person}{Can Zhang},
  \bibinfo{person}{Anthony Bonifonte}, \bibinfo{person}{Anne~C Spaulding},
  {and} \bibinfo{person}{Jagpreet Chhatwal}.} \bibinfo{year}{2019}\natexlab{}.
\newblock \showarticletitle{Prioritizing hepatitis C treatment in US prisons}.
\newblock \bibinfo{journal}{\emph{Operations Research}} \bibinfo{volume}{67},
  \bibinfo{number}{3} (\bibinfo{year}{2019}), \bibinfo{pages}{853--873}.
\newblock


\bibitem[\protect\citeauthoryear{Bagheri and Scaglione}{Bagheri and
  Scaglione}{2015}]%
        {bagheri2015restless}
\bibfield{author}{\bibinfo{person}{Saeed Bagheri} {and} \bibinfo{person}{Anna
  Scaglione}.} \bibinfo{year}{2015}\natexlab{}.
\newblock \showarticletitle{The restless multi-armed bandit formulation of the
  cognitive compressive sensing problem}.
\newblock \bibinfo{journal}{\emph{IEEE Transactions on Signal Processing}}
  \bibinfo{volume}{63}, \bibinfo{number}{5} (\bibinfo{year}{2015}),
  \bibinfo{pages}{1183--1198}.
\newblock


\bibitem[\protect\citeauthoryear{Biswas, Aggarwal, Varakantham, and
  Tambe}{Biswas et~al\mbox{.}}{2021}]%
        {biswas2021learn}
\bibfield{author}{\bibinfo{person}{Arpita Biswas}, \bibinfo{person}{Gaurav
  Aggarwal}, \bibinfo{person}{Pradeep Varakantham}, {and}
  \bibinfo{person}{Milind Tambe}.} \bibinfo{year}{2021}\natexlab{}.
\newblock \showarticletitle{Learn to intervene: An adaptive learning policy for
  restless bandits in application to preventive healthcare}. In
  \bibinfo{booktitle}{\emph{Proc. of IJCAI}}.
\newblock


\bibitem[\protect\citeauthoryear{Boehmer, Zhao, Xiong, Rodriguez-Diaz, Cibrian,
  Ngonzi, Boatin, and Tambe}{Boehmer et~al\mbox{.}}{2024}]%
        {boehmer2024optimizingvitalsignmonitoring}
\bibfield{author}{\bibinfo{person}{Niclas Boehmer}, \bibinfo{person}{Yunfan
  Zhao}, \bibinfo{person}{Guojun Xiong}, \bibinfo{person}{Paula
  Rodriguez-Diaz}, \bibinfo{person}{Paola Del~Cueto Cibrian},
  \bibinfo{person}{Joseph Ngonzi}, \bibinfo{person}{Adeline Boatin}, {and}
  \bibinfo{person}{Milind Tambe}.} \bibinfo{year}{2024}\natexlab{}.
\newblock \bibinfo{title}{Optimizing Vital Sign Monitoring in
  Resource-Constrained Maternal Care: An RL-Based Restless Bandit Approach}.
\newblock
\newblock
\showeprint[arxiv]{2410.08377}~[cs.AI]
\urldef\tempurl%
\url{https://arxiv.org/abs/2410.08377}
\showURL{%
\tempurl}


\bibitem[\protect\citeauthoryear{Borkar and Chadha}{Borkar and Chadha}{2018}]%
        {borkar2018reinforcement}
\bibfield{author}{\bibinfo{person}{Vivek~S Borkar} {and} \bibinfo{person}{Karan
  Chadha}.} \bibinfo{year}{2018}\natexlab{}.
\newblock \showarticletitle{A reinforcement learning algorithm for restless
  bandits}. In \bibinfo{booktitle}{\emph{2018 Indian Control Conference
  (ICC)}}. IEEE, \bibinfo{pages}{89--94}.
\newblock


\bibitem[\protect\citeauthoryear{Brown and Smith}{Brown and Smith}{2020}]%
        {brown2020index}
\bibfield{author}{\bibinfo{person}{David~B Brown} {and}
  \bibinfo{person}{James~E Smith}.} \bibinfo{year}{2020}\natexlab{}.
\newblock \showarticletitle{{Index Policies and Performance Bounds for Dynamic
  Selection Problems}}.
\newblock \bibinfo{journal}{\emph{Management Science}} \bibinfo{volume}{66},
  \bibinfo{number}{7} (\bibinfo{year}{2020}), \bibinfo{pages}{3029--3050}.
\newblock


\bibitem[\protect\citeauthoryear{Brown and Zhang}{Brown and Zhang}{2023}]%
        {brown2023fluid}
\bibfield{author}{\bibinfo{person}{David~B Brown} {and}
  \bibinfo{person}{Jingwei Zhang}.} \bibinfo{year}{2023}\natexlab{}.
\newblock \showarticletitle{Fluid policies, reoptimization, and performance
  guarantees in dynamic resource allocation}.
\newblock \bibinfo{journal}{\emph{Operations Research}} (\bibinfo{year}{2023}).
\newblock


\bibitem[\protect\citeauthoryear{Cohen, Zhao, and Scaglione}{Cohen
  et~al\mbox{.}}{2014}]%
        {cohen2014restless}
\bibfield{author}{\bibinfo{person}{Kobi Cohen}, \bibinfo{person}{Qing Zhao},
  {and} \bibinfo{person}{Anna Scaglione}.} \bibinfo{year}{2014}\natexlab{}.
\newblock \showarticletitle{Restless multi-armed bandits under time-varying
  activation constraints for dynamic spectrum access}. In
  \bibinfo{booktitle}{\emph{2014 48th Asilomar Conference on Signals, Systems
  and Computers}}. IEEE, \bibinfo{pages}{1575--1578}.
\newblock


\bibitem[\protect\citeauthoryear{Dai, Gai, Krishnamachari, and Zhao}{Dai
  et~al\mbox{.}}{2011}]%
        {dai2011non}
\bibfield{author}{\bibinfo{person}{Wenhan Dai}, \bibinfo{person}{Yi Gai},
  \bibinfo{person}{Bhaskar Krishnamachari}, {and} \bibinfo{person}{Qing Zhao}.}
  \bibinfo{year}{2011}\natexlab{}.
\newblock \showarticletitle{{The Non-Bayesian Restless Multi-Armed Bandit: A
  Case of Near-Logarithmic Regret}}. In \bibinfo{booktitle}{\emph{Proc. of IEEE
  ICASSP}}.
\newblock


\bibitem[\protect\citeauthoryear{Esposito and Mondelli}{Esposito and
  Mondelli}{2024}]%
        {esposito2024concentration}
\bibfield{author}{\bibinfo{person}{Amedeo~Roberto Esposito} {and}
  \bibinfo{person}{Marco Mondelli}.} \bibinfo{year}{2024}\natexlab{}.
\newblock \showarticletitle{Concentration without independence via information
  measures}.
\newblock \bibinfo{journal}{\emph{IEEE Transactions on Information Theory}}
  (\bibinfo{year}{2024}).
\newblock


\bibitem[\protect\citeauthoryear{Fu, Nazarathy, Moka, and Taylor}{Fu
  et~al\mbox{.}}{2019}]%
        {fu2019towards}
\bibfield{author}{\bibinfo{person}{Jing Fu}, \bibinfo{person}{Yoni Nazarathy},
  \bibinfo{person}{Sarat Moka}, {and} \bibinfo{person}{Peter~G Taylor}.}
  \bibinfo{year}{2019}\natexlab{}.
\newblock \showarticletitle{Towards q-learning the whittle index for restless
  bandits}. In \bibinfo{booktitle}{\emph{2019 Australian \& New Zealand Control
  Conference (ANZCC)}}. IEEE, \bibinfo{pages}{249--254}.
\newblock


\bibitem[\protect\citeauthoryear{Gast, Gaujal, and Yan}{Gast
  et~al\mbox{.}}{2023}]%
        {gast2023linear}
\bibfield{author}{\bibinfo{person}{Nicolas Gast}, \bibinfo{person}{Bruno
  Gaujal}, {and} \bibinfo{person}{Chen Yan}.} \bibinfo{year}{2023}\natexlab{}.
\newblock \showarticletitle{Linear program-based policies for restless bandits:
  Necessary and sufficient conditions for (exponentially fast) asymptotic
  optimality}.
\newblock \bibinfo{journal}{\emph{Mathematics of Operations Research}}
  (\bibinfo{year}{2023}).
\newblock


\bibitem[\protect\citeauthoryear{Ghosh, Nagaraj, Jain, and Tambe}{Ghosh
  et~al\mbox{.}}{2022}]%
        {ghosh2022indexability}
\bibfield{author}{\bibinfo{person}{Abheek Ghosh}, \bibinfo{person}{Dheeraj
  Nagaraj}, \bibinfo{person}{Manish Jain}, {and} \bibinfo{person}{Milind
  Tambe}.} \bibinfo{year}{2022}\natexlab{}.
\newblock \showarticletitle{Indexability is not enough for whittle: Improved,
  near-optimal algorithms for restless bandits}.
\newblock \bibinfo{journal}{\emph{arXiv preprint arXiv:2211.00112}}
  (\bibinfo{year}{2022}).
\newblock


\bibitem[\protect\citeauthoryear{Herlihy, Prins, Srinivasan, and
  Dickerson}{Herlihy et~al\mbox{.}}{2023}]%
        {herlihy2023planning}
\bibfield{author}{\bibinfo{person}{Christine Herlihy}, \bibinfo{person}{Aviva
  Prins}, \bibinfo{person}{Aravind Srinivasan}, {and} \bibinfo{person}{John~P
  Dickerson}.} \bibinfo{year}{2023}\natexlab{}.
\newblock \showarticletitle{Planning to fairly allocate: Probabilistic fairness
  in the restless bandit setting}. In \bibinfo{booktitle}{\emph{Proceedings of
  the 29th ACM SIGKDD Conference on Knowledge Discovery and Data Mining}}.
  \bibinfo{pages}{732--740}.
\newblock


\bibitem[\protect\citeauthoryear{Hoeffding}{Hoeffding}{1994}]%
        {hoeffding1994probability}
\bibfield{author}{\bibinfo{person}{Wassily Hoeffding}.}
  \bibinfo{year}{1994}\natexlab{}.
\newblock \showarticletitle{Probability inequalities for sums of bounded random
  variables}.
\newblock \bibinfo{journal}{\emph{The collected works of Wassily Hoeffding}}
  (\bibinfo{year}{1994}), \bibinfo{pages}{409--426}.
\newblock


\bibitem[\protect\citeauthoryear{Hu and Frazier}{Hu and Frazier}{2017}]%
        {hu2017asymptotically}
\bibfield{author}{\bibinfo{person}{Weici Hu} {and} \bibinfo{person}{Peter
  Frazier}.} \bibinfo{year}{2017}\natexlab{}.
\newblock \showarticletitle{{An Asymptotically Optimal Index Policy for
  Finite-Horizon Restless Bandits}}.
\newblock \bibinfo{journal}{\emph{arXiv preprint arXiv:1707.00205}}
  (\bibinfo{year}{2017}).
\newblock


\bibitem[\protect\citeauthoryear{Jung and Tewari}{Jung and Tewari}{2019}]%
        {jung2019regret}
\bibfield{author}{\bibinfo{person}{Young~Hun Jung} {and} \bibinfo{person}{Ambuj
  Tewari}.} \bibinfo{year}{2019}\natexlab{}.
\newblock \showarticletitle{Regret bounds for thompson sampling in episodic
  restless bandit problems}.
\newblock \bibinfo{journal}{\emph{Advances in Neural Information Processing
  Systems (NeurIPS)}}  \bibinfo{volume}{32} (\bibinfo{year}{2019}).
\newblock


\bibitem[\protect\citeauthoryear{Killian, Biswas, Shah, and Tambe}{Killian
  et~al\mbox{.}}{2021a}]%
        {killian2021q}
\bibfield{author}{\bibinfo{person}{Jackson~A Killian}, \bibinfo{person}{Arpita
  Biswas}, \bibinfo{person}{Sanket Shah}, {and} \bibinfo{person}{Milind
  Tambe}.} \bibinfo{year}{2021}\natexlab{a}.
\newblock \showarticletitle{Q-Learning Lagrange Policies for Multi-Action
  Restless Bandits}. In \bibinfo{booktitle}{\emph{Proceedings of the 27th ACM
  SIGKDD Conference on Knowledge Discovery \& Data Mining}}.
  \bibinfo{pages}{871--881}.
\newblock


\bibitem[\protect\citeauthoryear{Killian, Perrault, and Tambe}{Killian
  et~al\mbox{.}}{2021b}]%
        {killian2021beyond}
\bibfield{author}{\bibinfo{person}{Jackson~A Killian}, \bibinfo{person}{Andrew
  Perrault}, {and} \bibinfo{person}{Milind Tambe}.}
  \bibinfo{year}{2021}\natexlab{b}.
\newblock \showarticletitle{{Beyond" To Act or Not to Act": Fast Lagrangian
  Approaches to General Multi-Action Restless Bandits}}. In
  \bibinfo{booktitle}{\emph{Proc.of AAMAS}}.
\newblock


\bibitem[\protect\citeauthoryear{Langendorf, Roederer, de~Pee, Brown, Doyon,
  Mamaty, Tour{\'e}, Manzo, and Grais}{Langendorf et~al\mbox{.}}{2014}]%
        {langendorf2014preventing}
\bibfield{author}{\bibinfo{person}{C{\'e}line Langendorf},
  \bibinfo{person}{Thomas Roederer}, \bibinfo{person}{Saskia de Pee},
  \bibinfo{person}{Denise Brown}, \bibinfo{person}{St{\'e}phane Doyon},
  \bibinfo{person}{Abdoul-Aziz Mamaty}, \bibinfo{person}{Lynda W-M Tour{\'e}},
  \bibinfo{person}{Mahamane~L Manzo}, {and} \bibinfo{person}{Rebecca~F Grais}.}
  \bibinfo{year}{2014}\natexlab{}.
\newblock \showarticletitle{Preventing acute malnutrition among young children
  in crises: a prospective intervention study in Niger}.
\newblock \bibinfo{journal}{\emph{PLoS medicine}} \bibinfo{volume}{11},
  \bibinfo{number}{9} (\bibinfo{year}{2014}), \bibinfo{pages}{e1001714}.
\newblock


\bibitem[\protect\citeauthoryear{Larra{\~n}aga, Ayesta, and
  Verloop}{Larra{\~n}aga et~al\mbox{.}}{2014}]%
        {larranaga2014index}
\bibfield{author}{\bibinfo{person}{Maialen Larra{\~n}aga},
  \bibinfo{person}{Urtzi Ayesta}, {and} \bibinfo{person}{Ina~Maria Verloop}.}
  \bibinfo{year}{2014}\natexlab{}.
\newblock \showarticletitle{{Index Policies for A Multi-Class Queue with Convex
  Holding Cost and Abandonments}}. In \bibinfo{booktitle}{\emph{Proc. of ACM
  Sigmetrics}}.
\newblock


\bibitem[\protect\citeauthoryear{Larrnaaga, Ayesta, and Verloop}{Larrnaaga
  et~al\mbox{.}}{2016}]%
        {larrnaaga2016dynamic}
\bibfield{author}{\bibinfo{person}{Maialen Larrnaaga}, \bibinfo{person}{Urtzi
  Ayesta}, {and} \bibinfo{person}{Ina~Maria Verloop}.}
  \bibinfo{year}{2016}\natexlab{}.
\newblock \showarticletitle{{Dynamic Control of Birth-and-Death Restless
  Bandits: Application to Resource-Allocation Problems}}.
\newblock \bibinfo{journal}{\emph{IEEE/ACM Transactions on Networking}}
  \bibinfo{volume}{24}, \bibinfo{number}{6} (\bibinfo{year}{2016}),
  \bibinfo{pages}{3812--3825}.
\newblock


\bibitem[\protect\citeauthoryear{Lee, Lavieri, and Volk}{Lee
  et~al\mbox{.}}{2019}]%
        {lee2019optimal}
\bibfield{author}{\bibinfo{person}{Elliot Lee}, \bibinfo{person}{Mariel~S
  Lavieri}, {and} \bibinfo{person}{Michael Volk}.}
  \bibinfo{year}{2019}\natexlab{}.
\newblock \showarticletitle{{Optimal Screening for Hepatocellular Carcinoma: A
  Restless Bandit Model}}.
\newblock \bibinfo{journal}{\emph{Manufacturing \& Service Operations
  Management}} \bibinfo{volume}{21}, \bibinfo{number}{1}
  (\bibinfo{year}{2019}), \bibinfo{pages}{198--212}.
\newblock


\bibitem[\protect\citeauthoryear{Li and Varakantham}{Li and
  Varakantham}{2022}]%
        {li2022towards}
\bibfield{author}{\bibinfo{person}{Dexun Li} {and} \bibinfo{person}{Pradeep
  Varakantham}.} \bibinfo{year}{2022}\natexlab{}.
\newblock \showarticletitle{Towards Soft Fairness in Restless Multi-Armed
  Bandits}.
\newblock \bibinfo{journal}{\emph{arXiv preprint arXiv:2207.13343}}
  (\bibinfo{year}{2022}).
\newblock


\bibitem[\protect\citeauthoryear{Liu, Liu, and Zhao}{Liu et~al\mbox{.}}{2012}]%
        {liu2012learning}
\bibfield{author}{\bibinfo{person}{Haoyang Liu}, \bibinfo{person}{Keqin Liu},
  {and} \bibinfo{person}{Qing Zhao}.} \bibinfo{year}{2012}\natexlab{}.
\newblock \showarticletitle{{Learning in A Changing World: Restless Multi-Armed
  Bandit with Unknown Dynamics}}.
\newblock \bibinfo{journal}{\emph{IEEE Transactions on Information Theory}}
  \bibinfo{volume}{59}, \bibinfo{number}{3} (\bibinfo{year}{2012}),
  \bibinfo{pages}{1902--1916}.
\newblock


\bibitem[\protect\citeauthoryear{Mate, Madaan, Taneja, Madhiwalla, Verma,
  Singh, Hegde, Varakantham, and Tambe}{Mate et~al\mbox{.}}{2022}]%
        {mate2022field}
\bibfield{author}{\bibinfo{person}{Aditya Mate}, \bibinfo{person}{Lovish
  Madaan}, \bibinfo{person}{Aparna Taneja}, \bibinfo{person}{Neha Madhiwalla},
  \bibinfo{person}{Shresth Verma}, \bibinfo{person}{Gargi Singh},
  \bibinfo{person}{Aparna Hegde}, \bibinfo{person}{Pradeep Varakantham}, {and}
  \bibinfo{person}{Milind Tambe}.} \bibinfo{year}{2022}\natexlab{}.
\newblock \showarticletitle{Field study in deploying restless multi-armed
  bandits: Assisting non-profits in improving maternal and child health}. In
  \bibinfo{booktitle}{\emph{Proceedings of the AAAI Conference on Artificial
  Intelligence}}, Vol.~\bibinfo{volume}{36}. \bibinfo{pages}{12017--12025}.
\newblock


\bibitem[\protect\citeauthoryear{Mate, Perrault, and Tambe}{Mate
  et~al\mbox{.}}{2021}]%
        {mate2021risk}
\bibfield{author}{\bibinfo{person}{Aditya Mate}, \bibinfo{person}{Andrew
  Perrault}, {and} \bibinfo{person}{Milind Tambe}.}
  \bibinfo{year}{2021}\natexlab{}.
\newblock \showarticletitle{{Risk-Aware Interventions in Public Health:
  Planning with Restless Multi-Armed Bandits}}. In
  \bibinfo{booktitle}{\emph{Proc.of AAMAS}}.
\newblock


\bibitem[\protect\citeauthoryear{McDiarmid et~al\mbox{.}}{McDiarmid
  et~al\mbox{.}}{1989}]%
        {mcdiarmid1989method}
\bibfield{author}{\bibinfo{person}{Colin McDiarmid} {et~al\mbox{.}}}
  \bibinfo{year}{1989}\natexlab{}.
\newblock \showarticletitle{On the method of bounded differences}.
\newblock \bibinfo{journal}{\emph{Surveys in combinatorics}}
  \bibinfo{volume}{141}, \bibinfo{number}{1} (\bibinfo{year}{1989}),
  \bibinfo{pages}{148--188}.
\newblock


\bibitem[\protect\citeauthoryear{Nakhleh, Ganji, Hsieh, Hou, Shakkottai,
  et~al\mbox{.}}{Nakhleh et~al\mbox{.}}{2021}]%
        {nakhleh2021neurwin}
\bibfield{author}{\bibinfo{person}{Khaled Nakhleh}, \bibinfo{person}{Santosh
  Ganji}, \bibinfo{person}{Ping-Chun Hsieh}, \bibinfo{person}{I Hou},
  \bibinfo{person}{Srinivas Shakkottai}, {et~al\mbox{.}}}
  \bibinfo{year}{2021}\natexlab{}.
\newblock \showarticletitle{NeurWIN: Neural Whittle index network for restless
  bandits via deep RL}.
\newblock \bibinfo{journal}{\emph{Advances in Neural Information Processing
  Systems (NeurIPS)}}  \bibinfo{volume}{34} (\bibinfo{year}{2021}),
  \bibinfo{pages}{828--839}.
\newblock


\bibitem[\protect\citeauthoryear{Ni{\~n}o-Mora}{Ni{\~n}o-Mora}{2023}]%
        {nino2023markovian}
\bibfield{author}{\bibinfo{person}{Jos{\'e} Ni{\~n}o-Mora}.}
  \bibinfo{year}{2023}\natexlab{}.
\newblock \showarticletitle{Markovian restless bandits and index policies: A
  review}.
\newblock \bibinfo{journal}{\emph{Mathematics}} \bibinfo{volume}{11},
  \bibinfo{number}{7} (\bibinfo{year}{2023}), \bibinfo{pages}{1639}.
\newblock


\bibitem[\protect\citeauthoryear{Ortner, Ryabko, Auer, and Munos}{Ortner
  et~al\mbox{.}}{2012}]%
        {ortner2012regret}
\bibfield{author}{\bibinfo{person}{Ronald Ortner}, \bibinfo{person}{Daniil
  Ryabko}, \bibinfo{person}{Peter Auer}, {and} \bibinfo{person}{R{\'e}mi
  Munos}.} \bibinfo{year}{2012}\natexlab{}.
\newblock \showarticletitle{{Regret Bounds for Restless Markov Bandits}}. In
  \bibinfo{booktitle}{\emph{Proc. of Algorithmic Learning Theory}}.
\newblock


\bibitem[\protect\citeauthoryear{Papadimitriou and Tsitsiklis}{Papadimitriou
  and Tsitsiklis}{1999}]%
        {papadimitriou1999complexity}
\bibfield{author}{\bibinfo{person}{Christos~H Papadimitriou} {and}
  \bibinfo{person}{John~N Tsitsiklis}.} \bibinfo{year}{1999}\natexlab{}.
\newblock \showarticletitle{The complexity of optimal queueing network
  control}.
\newblock \bibinfo{journal}{\emph{Mathematics of Operations Research}}
  \bibinfo{volume}{24}, \bibinfo{number}{2} (\bibinfo{year}{1999}),
  \bibinfo{pages}{293--305}.
\newblock


\bibitem[\protect\citeauthoryear{Puterman}{Puterman}{1994}]%
        {puterman1994markov}
\bibfield{author}{\bibinfo{person}{Martin~L Puterman}.}
  \bibinfo{year}{1994}\natexlab{}.
\newblock \bibinfo{booktitle}{\emph{{Markov Decision Processes: Discrete
  Stochastic Dynamic Programming}}}.
\newblock \bibinfo{publisher}{John Wiley \& Sons}.
\newblock


\bibitem[\protect\citeauthoryear{Tekin and Liu}{Tekin and Liu}{2011}]%
        {tekin2011adaptive}
\bibfield{author}{\bibinfo{person}{Cem Tekin} {and} \bibinfo{person}{Mingyan
  Liu}.} \bibinfo{year}{2011}\natexlab{}.
\newblock \showarticletitle{{Adaptive Learning of Uncontrolled Restless Bandits
  with Logarithmic Regret}}. In \bibinfo{booktitle}{\emph{Proc. of Allerton}}.
\newblock


\bibitem[\protect\citeauthoryear{Verloop}{Verloop}{2016}]%
        {verloop2016asymptotically}
\bibfield{author}{\bibinfo{person}{Ina~Maria Verloop}.}
  \bibinfo{year}{2016}\natexlab{}.
\newblock \showarticletitle{{Asymptotically Optimal Priority Policies for
  Indexable and Nonindexable Restless Bandits}}.
\newblock \bibinfo{journal}{\emph{The Annals of Applied Probability}}
  \bibinfo{volume}{26}, \bibinfo{number}{4} (\bibinfo{year}{2016}),
  \bibinfo{pages}{1947--1995}.
\newblock


\bibitem[\protect\citeauthoryear{Verma, Singh, Mate, Verma, Gorantla,
  Madhiwalla, Hegde, Thakkar, Jain, Tambe, et~al\mbox{.}}{Verma
  et~al\mbox{.}}{2023}]%
        {verma2023expanding}
\bibfield{author}{\bibinfo{person}{Shresth Verma}, \bibinfo{person}{Gargi
  Singh}, \bibinfo{person}{Aditya Mate}, \bibinfo{person}{Paritosh Verma},
  \bibinfo{person}{Sruthi Gorantla}, \bibinfo{person}{Neha Madhiwalla},
  \bibinfo{person}{Aparna Hegde}, \bibinfo{person}{Divy Thakkar},
  \bibinfo{person}{Manish Jain}, \bibinfo{person}{Milind Tambe},
  {et~al\mbox{.}}} \bibinfo{year}{2023}\natexlab{}.
\newblock \showarticletitle{Expanding impact of mobile health programs: SAHELI
  for maternal and child care}.
\newblock \bibinfo{journal}{\emph{AI Magazine}} \bibinfo{volume}{44},
  \bibinfo{number}{4} (\bibinfo{year}{2023}), \bibinfo{pages}{363--376}.
\newblock


\bibitem[\protect\citeauthoryear{Wang, Xiong, and Li}{Wang
  et~al\mbox{.}}{2024}]%
        {wang2024online}
\bibfield{author}{\bibinfo{person}{Shufan Wang}, \bibinfo{person}{Guojun
  Xiong}, {and} \bibinfo{person}{Jian Li}.} \bibinfo{year}{2024}\natexlab{}.
\newblock \showarticletitle{Online Restless Multi-Armed Bandits with Long-Term
  Fairness Constraints}. In \bibinfo{booktitle}{\emph{Proceedings of the AAAI
  Conference on Artificial Intelligence}}, Vol.~\bibinfo{volume}{38}.
  \bibinfo{pages}{15616--15624}.
\newblock


\bibitem[\protect\citeauthoryear{Weber and Weiss}{Weber and Weiss}{1990}]%
        {weber1990index}
\bibfield{author}{\bibinfo{person}{Richard~R Weber} {and}
  \bibinfo{person}{Gideon Weiss}.} \bibinfo{year}{1990}\natexlab{}.
\newblock \showarticletitle{{On An Index Policy for Restless Bandits}}.
\newblock \bibinfo{journal}{\emph{Journal of applied probability}}
  (\bibinfo{year}{1990}), \bibinfo{pages}{637--648}.
\newblock


\bibitem[\protect\citeauthoryear{Whittle}{Whittle}{1988}]%
        {whittle1988restless}
\bibfield{author}{\bibinfo{person}{Peter Whittle}.}
  \bibinfo{year}{1988}\natexlab{}.
\newblock \showarticletitle{{Restless Bandits: Activity Allocation in A
  Changing World}}.
\newblock \bibinfo{journal}{\emph{Journal of applied probability}}
  (\bibinfo{year}{1988}), \bibinfo{pages}{287--298}.
\newblock


\bibitem[\protect\citeauthoryear{Xiong and Li}{Xiong and Li}{2023}]%
        {xiong2023finite}
\bibfield{author}{\bibinfo{person}{Guojun Xiong} {and} \bibinfo{person}{Jian
  Li}.} \bibinfo{year}{2023}\natexlab{}.
\newblock \showarticletitle{Finite-time analysis of whittle index based
  Q-learning for restless multi-armed bandits with neural network function
  approximation}.
\newblock \bibinfo{journal}{\emph{Advances in Neural Information Processing
  Systems}}  \bibinfo{volume}{36} (\bibinfo{year}{2023}),
  \bibinfo{pages}{29048--29073}.
\newblock


\bibitem[\protect\citeauthoryear{Xiong, Li, and Singh}{Xiong
  et~al\mbox{.}}{2022a}]%
        {xiong2021reinforcement}
\bibfield{author}{\bibinfo{person}{Guojun Xiong}, \bibinfo{person}{Jian Li},
  {and} \bibinfo{person}{Rahul Singh}.} \bibinfo{year}{2022}\natexlab{a}.
\newblock \showarticletitle{Reinforcement learning augmented asymptotically
  optimal index policy for finite-horizon restless bandits}. In
  \bibinfo{booktitle}{\emph{Proceedings of the AAAI Conference on Artificial
  Intelligence}}, Vol.~\bibinfo{volume}{36}. \bibinfo{pages}{8726--8734}.
\newblock


\bibitem[\protect\citeauthoryear{Xiong, Qin, Li, Singh, and Li}{Xiong
  et~al\mbox{.}}{2022b}]%
        {xiong2022indexwireless}
\bibfield{author}{\bibinfo{person}{Guojun Xiong}, \bibinfo{person}{Xudong Qin},
  \bibinfo{person}{Bin Li}, \bibinfo{person}{Rahul Singh}, {and}
  \bibinfo{person}{Jian Li}.} \bibinfo{year}{2022}\natexlab{b}.
\newblock \showarticletitle{{Index-aware Reinforcement Learning for Adaptive
  Video Streaming at the Wireless Edge}}. In \bibinfo{booktitle}{\emph{Proc. of
  ACM MobiHoc}}.
\newblock


\bibitem[\protect\citeauthoryear{Xiong, Wang, and Li}{Xiong
  et~al\mbox{.}}{2022c}]%
        {xiong2022Nips}
\bibfield{author}{\bibinfo{person}{Guojun Xiong}, \bibinfo{person}{Shufan
  Wang}, {and} \bibinfo{person}{Jian Li}.} \bibinfo{year}{2022}\natexlab{c}.
\newblock \showarticletitle{Learning infinite-horizon average-reward restless
  multi-action bandits via index awareness}.
\newblock \bibinfo{journal}{\emph{Advances in Neural Information Processing
  Systems}}  \bibinfo{volume}{35} (\bibinfo{year}{2022}),
  \bibinfo{pages}{17911--17925}.
\newblock


\bibitem[\protect\citeauthoryear{Xiong, Wang, Yan, and Li}{Xiong
  et~al\mbox{.}}{2022d}]%
        {xiong2022reinforcementcache}
\bibfield{author}{\bibinfo{person}{Guojun Xiong}, \bibinfo{person}{Shufan
  Wang}, \bibinfo{person}{Gang Yan}, {and} \bibinfo{person}{Jian Li}.}
  \bibinfo{year}{2022}\natexlab{d}.
\newblock \showarticletitle{{Reinforcement Learning for Dynamic Dimensioning of
  Cloud Caches: A Restless Bandit Approach}}. In
  \bibinfo{booktitle}{\emph{Proc. of IEEE INFOCOM}}.
\newblock


\bibitem[\protect\citeauthoryear{Yu, Xu, and Tong}{Yu et~al\mbox{.}}{2018}]%
        {yu2018deadline}
\bibfield{author}{\bibinfo{person}{Zhe Yu}, \bibinfo{person}{Yunjian Xu}, {and}
  \bibinfo{person}{Lang Tong}.} \bibinfo{year}{2018}\natexlab{}.
\newblock \showarticletitle{{Deadline Scheduling as Restless Bandits}}.
\newblock \bibinfo{journal}{\emph{IEEE Trans. Automat. Control}}
  \bibinfo{volume}{63}, \bibinfo{number}{8} (\bibinfo{year}{2018}),
  \bibinfo{pages}{2343--2358}.
\newblock


\bibitem[\protect\citeauthoryear{Zayas-Cab{\'a}n, Jasin, and
  Wang}{Zayas-Cab{\'a}n et~al\mbox{.}}{2019}]%
        {zayas2019asymptotically}
\bibfield{author}{\bibinfo{person}{Gabriel Zayas-Cab{\'a}n},
  \bibinfo{person}{Stefanus Jasin}, {and} \bibinfo{person}{Guihua Wang}.}
  \bibinfo{year}{2019}\natexlab{}.
\newblock \showarticletitle{{An Asymptotically Optimal Heuristic for General
  Nonstationary Finite-Horizon Restless Multi-Armed, Multi-Action Bandits}}.
\newblock \bibinfo{journal}{\emph{Advances in Applied Probability}}
  \bibinfo{volume}{51}, \bibinfo{number}{3} (\bibinfo{year}{2019}),
  \bibinfo{pages}{745--772}.
\newblock


\bibitem[\protect\citeauthoryear{Zhang and Frazier}{Zhang and Frazier}{2021}]%
        {zhang2021restless}
\bibfield{author}{\bibinfo{person}{Xiangyu Zhang} {and}
  \bibinfo{person}{Peter~I Frazier}.} \bibinfo{year}{2021}\natexlab{}.
\newblock \showarticletitle{Restless Bandits with Many Arms: Beating the
  Central Limit Theorem}.
\newblock \bibinfo{journal}{\emph{arXiv preprint arXiv:2107.11911}}
  (\bibinfo{year}{2021}).
\newblock


\bibitem[\protect\citeauthoryear{Zhao, Behari, Hughes, Zhang, Nagaraj, Tuyls,
  Taneja, and Tambe}{Zhao et~al\mbox{.}}{2024}]%
        {zhao2024towards}
\bibfield{author}{\bibinfo{person}{Yunfan Zhao}, \bibinfo{person}{Nikhil
  Behari}, \bibinfo{person}{Edward Hughes}, \bibinfo{person}{Edwin Zhang},
  \bibinfo{person}{Dheeraj Nagaraj}, \bibinfo{person}{Karl Tuyls},
  \bibinfo{person}{Aparna Taneja}, {and} \bibinfo{person}{Milind Tambe}.}
  \bibinfo{year}{2024}\natexlab{}.
\newblock \showarticletitle{Towards a Pretrained Model for Restless Bandits via
  Multi-arm Generalization}. IJCAI.
\newblock


\end{thebibliography}

\newpage
%%%%%%%%%%%%%%%%%%%%%%%%%%%%%%%%%%%%%%%%%%%%%%%%%%%%%%%%%%%%%%%%%%%%%%%%
\appendix
\section{Related Work}
The RMAB problem, introduced by Whittle~\cite{whittle1988restless}, is \textit{PSPACE-hard}~\cite{papadimitriou1999complexity}, and exact solutions are often infeasible to find for large instances.  
Whittle's seminal work on RMABs proposed the \textit{Whittle index policy}, a heuristic based on the concept of \textit{indexability}. Whittle's index provides a method to decouple the problem into simpler subproblems for each arm by assigning each state of an arm an index (Whittle index) that quantifies the "value" of activating that arm in a given state~\cite{whittle1988restless}.  The Whittle index is computationally efficient, especially for large-scale RMAB problems, and provides near-optimal performance in many cases when the arms are \textit{indexable}.
However, a major challenge in applying Whittle’s index is determining the \textit{indexability} of arms, which is not guaranteed in all RMAB instances.

 The aforementioned challenges have led to development of several approximation algorithms to tackle RMABs. They can be categorized into two main approaches \cite{nino2023markovian}: \textit{offline index-based approaches} and \textit{online learning-based approaches}:

\begin{itemize}
   \item \textbf{Offline  Index-based Approaches:} One of the most popular offline approaches for RMABs is the Lagrangian-based method following Whittle’s foundational work, which provides approximate solutions by relaxing the problem constraints. These methods break down complex optimization problems into simpler subproblems that are easier to solve. Subsequent literature \cite{hu2017asymptotically, zayas2019asymptotically, brown2020index} has explored the finite-horizon restless bandit problem using Lagrangian relaxations. Unlike Whittle’s original formulation, these approaches employ time-dependent Lagrange multipliers to account for the nonstationary nature of finite-horizon problems. This technique offers promising performance guarantees and empirical results without relying on the indexability condition. Later,  \cite{killian2021beyond} extended the Lagrangian-based method to handle multi-action cases, further broadening its applicability.

  The LP relaxation method is another prominent technique used in RMABs, where the original combinatorial problem is relaxed into an LP formulation. This method computes a fractional solution that is then rounded to an integer solution, often yielding near-optimal results. LP-based approaches are particularly useful in large-scale RMAB problems, such as resource allocation in wireless networks or healthcare settings, where exact solutions may be computationally prohibitive. Several studies have explored the effectiveness of LP relaxation methods in RMABs, including \cite{verloop2016asymptotically,zhang2021restless,ghosh2022indexability, brown2023fluid, gast2023linear}, showcasing their utility in approximating solutions for large-scale, complex problems.

     \item \textbf{Online Learning-based Approaches}: The online RMAB setting, where the underlying MDPs are unknown, has been gaining attention, e.g., \cite{dai2011non,liu2012learning,ortner2012regret,jung2019regret}.  However, these methods do not exploit the special structure available in the problem and contend directly with an extremely high dimensional state-action space yielding the algorithms to be too slow to be useful.  Recently,  RL based algorithms have been developed \cite{borkar2018reinforcement,biswas2021learn,killian2021q,xiong2022reinforcementcache,xiong2022indexwireless,avrachenkov2022whittle, zhao2024towards}, 
to explore the problem structure through index policies.  
For instance, \cite{ avrachenkov2022whittle} proposed a Q-learning algorithm for Whittle index under the discounted setting, which lacks of convergence guarantees.  \cite{biswas2021learn} approximated Whittle index using the difference of Q functions for any state, which is not guaranteed to converge to the true Whittle index in general scenarios. 
\cite{xiong2021reinforcement} designed model-based low-complexity policy but is constrained to either a specific Markovian model or depends on a simulator for a finite-horizon setting which cannot be directly applied here.

\end{itemize}

\section{Missing Proofs}
\subsection{Proof Skecth of Proposition \ref{prop:indexability}}
We focus on a single MDP of one patient in Example \ref{example1}. Since the set of feasible policies $\Pi$ is non-empty, there exists a stationary policy $\pi^*$ that optimally solves this MDP. To simplify the presenttaion, we consider a general state space rather than $3$ states in Example \ref{example1}.
Let $S^*=\max\{S\in\{0, 1,\cdots\}: a^{\pi^*}(S)=0\}$, where we use the superscript $\pi^*$ to denote actions under policy $\pi^*$ and let $a_n^{\pi^*}(s)$ be the action at state $s$ under policy $\pi^*$. By definition, we have $a^{\pi^*}(s)=1$, $\forall s>S^*.$

Given the transition probability under action $0$, all states below $S^*$  are transient, implying the stationary probability in state $s$ under policy $\pi^*$ being $0$, i.e., $\phi^{\pi^*}(s)=0$, $\forall s<S^*$. Hence, the average reward under the optimal policy $\pi^*$ with Lagrangian multiplier $\lambda$ reduces to
\begin{align*}
&\mathbb{E}_{\pi^*}[R(s)] -\lambda \mathbb{E}[\mathds{1}_{\{ a^{\pi^*}(s)=0\}}]=\sum_{s=0}^{S^*-1}R(s) \phi^{\pi^*}(s)+S^*\phi_n^{\pi^*}(S^*)\displaybreak[1]\\
&\qquad\qquad+\sum_{s=S^*+1}^{\infty} s \phi^{\pi^*}(s) - \lambda \sum_{s=0}^{S^*}\phi^{\pi^*}(s)\mathds{1}_{\{a^{\pi^*}(s)=0\}}\\
&\qquad\qquad=\sum_{s=R^*}^{\infty} s \phi_n^{\pi^*}(s) - \lambda\phi_n^{\pi^*}(R^*),
\end{align*}
i.e., the threshold-type policy optimally solves the single arm MDP.  For Birth-Death process, if the optimal policy of single MDP is of threshold-type, we can always show that the passive set monotonically increases as the Lagrangian multiplier $\lambda$ increases \cite{larranaga2014index,larrnaaga2016dynamic}. This is the definition of indexability of Whittle index policy \cite{whittle1988restless, weber1990index}.

\subsection{Proof Sketch of Theorem \ref{thm:asym_opt}}

For any policy $\pi$ derived from Algorithm~\ref{alg:Index}, the left-hand side of~(\ref{eq:Asym_opt}) is non-negative.
To prove \eqref{eq:Asym_opt}, it suffices to show that
\begin{align}\label{eq:reverse_direction}
    \lim_{\rho\rightarrow \infty} \frac{1}{\rho N}J^{\pi^{opt}}( \rho K, \rho N)\leq\lim_{\rho\rightarrow \infty} \frac{1}{\rho N}J^{\pi^{SPI}}(\rho K, \rho N).
\end{align}

Let $B_n(s; t)$ be the number of class $n$ users in state $s$ at time  $t$ and $D_n(s,1;t)$ be the number of class $n$ arms in state $s$ at time $t$ that are being pulled with action $a=1$.
The key is to show that $B_n(s; t)/\rho\rightarrow P_n(s;t)$ and $D_n(s,1; t)/\rho\rightarrow P_n(s;t)\chi_n^\star(s,1;t)$, respectively, as $\rho\rightarrow\infty$ almost surely. This leads to the fact that
\begin{align*}
   \lim_{\rho\rightarrow \infty} \frac{1}{\rho N}J^{\pi^{SPI}}(\rho K, \rho N)=\sum_{n=1}^{N}\sum_{t=1}^T\sum\limits_{(s,a)} \mu_n^\star(s,a;t)r_n(s,a),
\end{align*}
which is an upper bound of $\lim_{\rho\rightarrow \infty} \frac{1}{\rho N}J^{\pi^{opt}}(\rho K, \rho N)$ according to Proposition \ref{prop:upperbound1}. This verifies the inequality in \eqref{eq:reverse_direction}.
For detailed proof, please refer to \cite{zhang2021restless,gast2023linear,brown2023fluid,xiong2021reinforcement}.

\subsection{Proof of Theorem \ref{thm:non-asymptotic}}
Similar as showing the asymptotic optimality, we first expand the system by a factor $\rho\in\mathbb{R}^+$:

\begin{itemize}
    \item \textbf{Arm Replication:} Each arm $i$ is replicated $\rho$ times.
    \item \textbf{Total Arms:} The total number of arms is $N \rho$.
    \item \textbf{Budget Scaling:} The budget is scaled accordingly, so at each time $t$, at most $K' = K \rho$ arms can be activated.
\end{itemize}
We next define two policies and their expected total rewards:

\textbf{Optimal Policy $\pi^{Opt}$:} The policy that maximizes the expected total reward over the time horizon $T$. Hence, the expected total rewards
under optimal policy $\pi^{Opt}$ is defined as:
\begin{align}
  J^{\pi^{Opt}}( \rho K, \rho N) = \mathbb{E}_{\pi^{Opt}} \left[ \sum_{t=1}^T \sum_{i=1}^{N}\sum_{j=1}^\rho r_{i,j}(s_{i,j}(t), a_{i,j}(t)) \right] .
\end{align}

\textbf{\texttt{SPI}  Policy $\pi^{SPI}$:} A policy based on the indices derived from the occupancy measures.
Similarly, for the index policy $\pi^{SPI}$, we have:
\begin{align}
    J^{\pi^{SPI}}( \rho K, \rho N) = \mathbb{E}_{\pi^{SPI}} \left[ \sum_{t=1}^T \sum_{i=1}^{N}\sum_{j=1}^\rho r_{i,j}(s_{i,j}(t), a_{i,j}(t)) \right] .
\end{align}
Next,
we define the \textbf{expected total reward difference} between the optimal policy and the index policy as:
\begin{align}
    \mathbb{E}[\Delta V] = J^{\pi^{Opt}}( \rho K, \rho N) - J^{\pi^{SPI}}( \rho K, \rho N).
\end{align}
Our goal is to bound $\Delta V$ and analyze how it behaves as $\rho$ increases.

\subsubsection{ Decomposing $\mathbb{E}[\Delta V]$ into Individual Contributions}

We can decompose $\mathbb{E}[\Delta V]$ as the sum over all arm classes:
\begin{align}
    \mathbb{E}[\Delta V] = \sum_{i=1}^{N } \Delta V_i,
\end{align}
where $\Delta V_i = V_i^{{Opt}} - V_i^{{Index}}$ is the expected reward difference for arm class $i$ with total $\rho$ arms.

\subsubsection{ Bounding Individual Reward Differences}

For each arm class $i$, define the random variable:
\begin{align}
    X_i = \sum_{t=1}^T \sum_{j=1}^\rho\left( R_{i,j}^{{Opt}}(t) - R_{i,j}^{{SPI}}(t) \right),
\end{align}
where $R_{i,j}^{{Opt}}(t)$ is the reward received from arm $j$ in class $i$ at time $t$ under the optimal policy, and
 $R_{i,j}^{{SPI}}(t)$ is the reward received from arm $j$ of class $i$ at time $t$ under the \texttt{SPI} policy. We also define  $R_{i}^{{Index}}(t):=\sum_{j=1}^\rho R_{i,j}^{{Index}}(t)$ for notationla simplicity.    Since the rewards are bounded ($|r_i(s, a)| \leq r_{\max}$), the difference at each time step is bounded:
\begin{align}
  | R_{i,j}^{{Opt}}(t) - R_{i,j}^{{SPI}}(t) | \leq 2 r_{\max}.
\end{align}
Therefore, we have the following bounded differences:
\begin{align*}
    \sum_{t=1}^T| R_{i,j}^{{Opt}}(t) - R_{i,j}^{{SPI}}(t) | \leq 2 r_{\max} T.
\end{align*}

\subsubsection{ Addressing Dependence Due to Budget Constraint}
The dependency between the arms is captured by the joint distribution of the rewards:
$ \mathbb{P}(R_{1,1}, R_{1,2}, \dots, R_{N,\rho}).$
If all arms were independent, the joint distribution would factorize as:
\begin{align*}
    \mathbb{P}_{\text{prod}}(R_{1,1}, R_{1,2}, \dots, R_{N,\rho}) = \prod_{i=1}^{N} \prod_{j=1}^{\rho} \mathbb{P}(R_{i,j}).
\end{align*}
However, since there are dependencies, we use the \textbf{Hellinger integral} \cite{esposito2024concentration} to measure the "distance" between the actual joint distribution $\mathbb{P}$ and the product of the marginals $\mathbb{P}_{\text{prod}}$, which is defined as:
\begin{align*}
   H_{\alpha}(\mathbb{P} \| \mathbb{P}_{\text{prod}}) = \int \left( \frac{d\mathbb{P}}{d\mathbb{P}_{\text{prod}}} \right)^{\alpha} d\mathbb{P}_{\text{prod}},
\end{align*}
with $\alpha $ being a parameter that controls the trade-off between how we treat the dependency. For our setting with $N\rho $ arms, we compute the Hellinger integral over the entire joint distribution:
\begin{align*}
    H_{\alpha}(\mathbb{P} \| \mathbb{P}_{\text{prod}}) &= \int_{R_{1,1},\dots,R_{N,\rho}} \left( \frac{\mathbb{P}(R_{1,1}, \dots, R_{N,\rho})}{\mathbb{P}_{\text{prod}}(R_{1,1}, \dots, R_{N,\rho})} \right)^{\alpha} \\
    &\cdot P_{\text{prod}}(R_{1,1}, \dots, R_{N,\rho}) dR.
\end{align*}
This integral quantifies how much the arms deviate from independence. If the arms are fully independent, $ \mathbb{P} = \mathbb{P}_{\text{prod}} $ and the Hellinger integral becomes 1. For dependent arms, $H_{\alpha}(\mathbb{P} \| \mathbb{P}_{\text{prod}}) > 1$.

\begin{lemma}[McDiarmid's Inequality\cite{mcdiarmid1989method}]\label{lemma:McDiarmid}
Formally, if~ $X_1, X_2, \dots, X_{N\rho}$ are  random variables and $f$ is a function satisfying:
\begin{align*}
    |f(x_1, x_2, \dots, x_{N,\rho}) - f(x_{1,1}, x_{1,2}, \dots, x'_i, \dots, x_{N,\rho})| \leq c_{i,j}, \quad \forall i,j,
\end{align*}
where $x'_{i,j}$ is any possible value of $X_{i,j}$, then for any $\epsilon > 0$:
\begin{align}
    &\Pr\left( f(X_{1,1}, X_{1,2}, \dots, X_{N,\rho}) - \mathbb{E}[f(X_{1,1}, X_{1,2}, \dots, X_{N,\rho})] \geq \epsilon \right) \nonumber\\
    &\qquad\qquad\qquad\qquad\leq 2 H_{\alpha}^{1/\alpha}(\mathbb{P} \| \mathbb{P}_{\text{prod}}) \exp\left( -\frac{2\epsilon^2}{\sum_{i=1}^{N\rho} c_{i,j}^2} \right).
\end{align}
\end{lemma}
Lemma \ref{lemma:McDiarmid} provides a concentration bound for functions of independent random variables, assuming that the function does not change too much when any single variable is altered. As the number of the same class arms $\rho$ increases, the dependency between arms can decrease, especially if the dependencies are local (e.g., within a class or due to resource constraints). Hence, Hellinger integral scales as $ \rho^{-\beta} $, where $ \beta \geq 0 $. This gives:
\begin{align*}
    H_{\alpha}(\mathbb{P} \| \mathbb{P}_{\text{prod}}) = \rho^{-\beta}.
\end{align*}

Therefore, based on Lemma \ref{lemma:McDiarmid} and the Hellinger integral, we have the final concentration inequality incorporating the dependency correction is:
\begin{align}\label{eq:DC}
   Pr \left( \left|\sum_{i=1}^N X_i - \mathbb{E}[\Delta V]\right| \geq \epsilon \right) \leq 2 \rho^{-\beta/\alpha} \exp\left( - \frac{2\epsilon^2}{N\rho c^2} \right).
\end{align}

\begin{lemma}
With probability at least $1-\frac{1}{\rho}$, the following inequality holds
\begin{align}
    \sum_{i=1}^N X_i\leq \mathbb{E}[\Delta V]+\sqrt{\frac{N\rho c^2\ln 2\rho}{2}}.
\end{align}

\end{lemma}
\begin{proof}
This is a direct result from Lemma \ref{lemma:McDiarmid} and \eqref{eq:DC}.
\end{proof}
\subsubsection{ Using Total Variation Distance to Analyze $\mathbb{E}[\Delta V]$}

Let $d_i^{{Opt}}(t)$ and $d_i^{{SPI}}(t)$ be the state distributions of arm $i$ at time $t$ under the optimal and \texttt{SPI} policies, respectively.
The \textbf{total variation (TV) distance} between the two distributions is:
\begin{align}
    D_{\text{TV}}( d_i^{{Opt}}(t), d_i^{{SPI}}(t) ) = \frac{1}{2} \sum_{s \in \cS} \left| d_i^{{Opt}}(s, t) - d_i^{{SPI}}(s, t) \right|.
\end{align}
For the ease of expression, we abuse the notation usage and let $d_i^{{Policy}}(s, t) $ denote the probability of state $s$ at time $t$ for arm $i$ under a particular policy.
Hence, the expected reward difference for arm $i$ at time $t$ can be bounded by the following lemma.
\begin{lemma}
The following inequality holds
\begin{align}
    &\left| \mathbb{E} [ R_i^{{Opt}}(t) ] - \mathbb{E} [ R_i^{{SPI}}(t) ] \right| \nonumber\\
    &\leq 2r_{\max} \cdot D_{\text{TV}}( d_i^{{Opt}}(t), d_i^{{SPI}}(t) ) \nonumber\\
    &\qquad+ r_{\max} \cdot \left| \Pr [ a_i(t) = 1 |\pi_i^{Opt}] - \Pr [ a_i(t) = 1 |\pi_i^{SPI}] \right|,
\end{align}
where the $\Pr[ a_i(t) = 1 |\pi_i^{Opt}]$ and $\Pr [ a_i(t) = 1 |\pi_i^{SPI}]$ denote probabilities of activating arm  $i$  at time  $t$  under the optimal and index policies, respectively.
\end{lemma}

\begin{proof}
According to definitions, the expected rewards under the optimal and index policies can be expressed as:
\begin{align}
    \mathbb{E} [ R_i^{\text{Opt}}(t) ] = \sum_{s} d_i^{\text{Opt}}(s,t) \sum_{a} \pi_i^{\text{Opt}}(a | s) r_i(a, s),\\
    \mathbb{E} [ R_i^{\text{SPI}}(t) ] = \sum_{s} d_i^{\text{SPI}}(s,t) \sum_{a} \pi_i^{\text{SPI}}(a | s) r_i(a, s).
\end{align}
Hence, the absolute difference in expected rewards is:
\begin{align}\nonumber
    &\left| \mathbb{E} [ R_i^{\text{Opt}}(t) ] - \mathbb{E} [ R_i^{\text{SPI}}(t) ] \right|\\
    &\qquad\qquad=\Bigg| \sum_{s} d_i^{\text{Opt}}(s,t) \sum_{a} \pi_i^{\text{Opt}}(a | s) r_i(a, s) \nonumber\\
    &\qquad\qquad\quad- \sum_{s} d_i^{\text{SPI}}(s,t) \sum_{a} \pi_i^{\text{SPI}}(a | s) r_i(a, s) \Bigg|.
\end{align}
Using the triangle inequality, we split the difference into two parts:

\begin{align}\nonumber
&\Bigg| \sum_{s} d_i^{\text{Opt}}(s,t) \sum_{a} \pi_i^{\text{Opt}}(a | s) r_i(a, s) \\ \nonumber
&\qquad\qquad- \sum_{s} d_i^{\text{SPI}}(s,t) \sum_{a} \pi_i^{\text{SPI}}(a | s) r_i(a, s) \Bigg| \\ \nonumber
&\leq \left| \sum_{s} \left( d_i^{\text{Opt}}(s,t) - d_i^{\text{SPI}}(s,t) \right) \sum_{a} \pi_i^{\text{Opt}}(a | s) r_i(a, s) \right| \\
&\quad + \left| \sum_{s} d_i^{\text{SPI}}(s,t) \sum_{a} \left( \pi_i^{\text{Opt}}(a | s) - \pi_i^{\text{SPI}}(a | s) \right) r_i(a, s) \right|.
\end{align}

\textbf{Bounding the first term.}
Due to the fact that \( |r_i(a, s)| \leq r_{\max} \), we can bound the first term using the TV distance between the state distributions under the optimal and index policies.
\begin{align}\nonumber
  &\left| \sum_{s} \left( d_i^{\text{Opt}}(s,t) - d_i^{\text{SPI}}(s,t) \right) \sum_{a} \pi_i^{\text{Opt}}(a | s) r_i(a, s) \right| \\
  &\qquad\qquad\leq 2r_{\max} \cdot D_{\text{TV}}(d_i^{\text{Opt}}(t), d_i^{\text{Index}}(t)),
\end{align}
where the inequality comes from the fact that for any functions  $f(s)$  bounded by  $|f(s)| \leq c$, the difference in expectations under two distributions  $P$  and  $Q$  is bounded by  $c$  times the TV distance between  $P$  and  $Q$ :
\begin{align}
    \left| \mathbb{E}_P[f] - \mathbb{E}_Q[f] \right| \leq c \cdot D_{\text{TV}}(P, Q).
\end{align}
In our case:
 $f(s) = \sum_{a} \pi_i^{\text{Opt}}(a | s) r_i(a, s)$ and hence $|f(s)| \leq r_{\max}$ due to the fact
		$$   \sum_{a} \pi_i^{\text{Opt}}(a | s) \leq 1~  \text{and}~  |r_i(a, s)| \leq r_{\max}. $$

\textbf{Bounding the second term.}
Again, due to the fact that $|r_i(a, s)| \leq r_{\max}$, we have
\begin{align}\nonumber
    &\left| \sum_{s} d_i^{\text{SPI}}(s,t) \sum_{a} \left( \pi_i^{\text{Opt}}(a | s) - \pi_i^{\text{SPI}}(a | s) \right) r_i(a, s) \right| \\
    & \leq r_{\max} \cdot\left| \sum_{s} d_i^{\text{SPI}}(s,t) \sum_{a} \left( \pi_i^{\text{Opt}}(a | s) - \pi_i^{\text{SPI}}(a | s) \right)  \right|.
\end{align}
Taking expectation over  $d_i^{\text{SPI}}(s,t)$ we have
\begin{align}
    &\left| \sum_{s} d_i^{\text{SPI}}(s,t) \sum_{a} \left( \pi_i^{\text{Opt}}(a | s) - \pi_i^{\text{SPI}}(a | s) \right)\right| \nonumber\\
    &\qquad\leq  \mathbb{E}{s \sim d_i^{\text{SPI}}(t)} \left[ D_{\text{TV}}(\pi_i^{\text{Opt}}(\cdot | s), \pi_i^{\text{SPI}}(\cdot | s)) \right].
\end{align}
Since the TV distance between the action distributions is captured by  $D_{\text{TV}}(\pi_i^{\text{Opt}}(t), \pi_i^{\text{SPI}}(t)) $, which aggregates over states:
\begin{align}
   &\mathbb{E}{s \sim d_i^{\text{SPI}}(t)} \left[ D{\text{TV}}(\pi_i^{\text{Opt}}(\cdot | s), \pi_i^{\text{SPI}}(\cdot | s)) \right] \nonumber\\
   &= \left| \Pr [ a_i(t) = 1 |\pi_i^{Opt}] - \Pr [ a_i(t) = 1 |\pi_i^{SPI}] \right|.
\end{align}

\textbf{Combining the Bounds.}
Adding the bounds from the two terms:
\begin{align}
   & \left| \mathbb{E} [ R_i^{\text{Opt}}(t) ] - \mathbb{E} [ R_i^{\text{SPI}}(t) ] \right| \leq 2r_{\max} \cdot D_{\text{TV}}(d_i^{\text{Opt}}(t), d_i^{\text{SPI}}(t))\nonumber\\
    &\qquad+  r_{\max} \cdot \left| \Pr [ a_i(t) = 1 |\pi_i^{Opt}] - \Pr [ a_i(t) = 1 |\pi_i^{SPI}] \right|.
\end{align}

\end{proof}

We have the following Lemma regarding $D_{\text{TV}} ( d_i^{\text{Opt}}(t), d_i^{\text{SPI}}(t) )$.
\begin{lemma}\label{lemma:4}
The total variation distance between the state distributions at time $t+1$ can be related to that at time $t$:
  \begin{align*}
  D_{\text{TV}} &( d_i^{\text{Opt}}(t+1), d_i^{\text{SPI}}(t+1) )\nonumber\\
   &\leq 2 \left| \Pr[ a_i(t) = 1|\pi_i^{Opt} ] - \Pr[ a_i(t) = 1 |\pi_i^{SPI}] \right|.
\end{align*}

This is because the difference in state distributions at time $t+1$ arises from both the difference in state distributions at time $t$ and the difference in action probabilities.
\end{lemma}

\begin{proof}
The state distribution of arm $i$ at time $ t+1 $ under both policies can be expressed as a mixture of the state transitions based on the actions taken:
\begin{align}
    d_i^{\text{Opt}}(t+1) &= \Pr [ a_i(t) = 1 |\pi_i^{Opt}] \cdot T(1) \nonumber\\
    &+ \Pr [ a_i(t) = 0 |\pi_i^{Opt}] \cdot T(0),\nonumber\\
    d_i^{\text{SPI}}(t+1) &= \Pr [ a_i(t) = 1 |\pi_i^{SPI}] \cdot T(1) \nonumber\\
    &+ \Pr [ a_i(t) = 0 |\pi_i^{SPI}] \cdot T(0).
\end{align}
Here, $T(a) $ represents the state transition probability distribution given action $a$.
The TV distance between the state distributions under the two policies at time $t+1$ is:
\begin{align}\nonumber
&D_{\text{TV}}( d_i^{\text{Opt}}(t+1), d_i^{\text{SPI}}(t+1) ) \nonumber\\
&= \sup_{s} \left| \Pr [ a_i(t) = 1|\pi_i^{opt} ] \cdot T(1)(s) + \Pr [ a_i(t) = 0 |\pi_i^{opt}] \cdot T(0)(s) \right. \nonumber\\
&\quad \left. -  \Pr [ a_i(t) = 1|\pi_i^{SPI} ] \cdot T(1)(s) - \Pr [ a_i(t) = 0 |\pi_i^{SPI}] \cdot T(0)(s) \right| \nonumber\\
&= \sup_{s} \left| \left( \Pr [ a_i(t) = 1|\pi_i^{opt} ] - \Pr [ a_i(t) = 1|\pi_i^{SPI} ] \right) \cdot T(1)(s) \right. \nonumber\\
&\quad \left. + \left( \Pr[ a_i(t) = 0|\pi_i^{opt} ] - \Pr[ a_i(t) = 0|\pi_i^{SPI} ] \right) \cdot T(0)(s) \right|.
\end{align}
Using the triangle inequality, we can separate the terms:
\begin{align}\label{eq:46}\nonumber
&D_{\text{TV}}( d_i^{\text{Opt}}(t+1), d_i^{\text{SPI}}(t+1) ) \nonumber\\
&\leq \left| \Pr [ a_i(t) = 1|\pi_i^{opt} ] - \Pr [ a_i(t) = 1|\pi_i^{SPI} \right| \cdot \sup_{s} |T(1)(s)| \nonumber\\
&\quad + \left| \Pr [ a_i(t) = 0|\pi_i^{opt} ] - \Pr [ a_i(t) = 0|\pi_i^{SPI}\right| \cdot \sup_{s} |T(0)(s)| \nonumber\\
&\leq \left| \Pr [ a_i(t) = 1|\pi_i^{opt} ] - \Pr [ a_i(t) = 1|\pi_i^{SPI} \right| \cdot 1 \nonumber\\
&\quad + \left| \Pr [ a_i(t) = 0|\pi_i^{opt} ] - \Pr [ a_i(t) = 0|\pi_i^{SPI} \right| \cdot 1 \nonumber\\
&= \left| \Pr [ a_i(t) = 1|\pi_i^{opt} ] - \Pr [ a_i(t) = 1|\pi_i^{SPI} \right| \nonumber\\
&\qquad\qquad+ \left|  \Pr [ a_i(t) = 0|\pi_i^{opt} ] - \Pr [ a_i(t) = 0|\pi_i^{SPI}\right|.
\end{align}
Since the probabilities of all actions sum to one under any policy:
\begin{align}
    \Pr[ a_i(t) = 1|\pi_i^{opt} ] + \Pr[ a_i(t) = 0|\pi_i^{opt} ] = 1,\nonumber\\
    \Pr[ a_i(t) = 1|\pi_i^{SPI} ] + \Pr[ a_i(t) = 0|\pi_i^{SPI} ] = 1.
\end{align}
Hence we have
\begin{align}
   &\left|  \Pr [ a_i(t) = 1|\pi_i^{opt} ] - \Pr [ a_i(t) = 1|\pi_i^{SPI} \right| \nonumber\\
   &\qquad\qquad= \left|  \Pr [ a_i(t) = 0|\pi_i^{opt} ] - \Pr [ a_i(t) = 0|\pi_i^{SPI} \right|.
\end{align}
Substituting back to \eqref{eq:46} yields
\begin{align}
  D_{\text{TV}} &( d_i^{\text{Opt}}(t+1), d_i^{\text{SPI}}(t+1) )\nonumber\\
   &\leq 2 \left| \Pr[ a_i(t) = 1|\pi_i^{Opt} ] - \Pr[ a_i(t) = 1 |\pi_i^{SPI}] \right|.
\end{align}

\end{proof}

\subsubsection{Bounding $\left|\mathbb{E} [ R_i^{\text{Opt}}(t) ] - \mathbb{E} [ R_i^{\text{SPI}}(t) ] \right|$}

\begin{align}
   & \left| \mathbb{E} [ R_i^{\text{Opt}}(t) ] - \mathbb{E} [ R_i^{\text{SPI}}(t) ] \right| \leq 2r_{\max} \cdot D_{\text{TV}}(d_i^{\text{Opt}}(t), d_i^{\text{SPI}}(t))\nonumber\\
    &\qquad+  r_{\max} \cdot \left| \Pr [ a_i(t) = 1 |\pi_i^{Opt}] - \Pr [ a_i(t) = 1 |\pi_i^{SPI}] \right|\nonumber\\
    &\leq r_{\max} \cdot \left| \Pr [ a_i(t) = 1 |\pi_i^{Opt}] - \Pr [ a_i(t) = 1 |\pi_i^{SPI}] \right|\nonumber\\
    &-r_{\max} \cdot \left| \Pr [ a_i(t-1) = 1 |\pi_i^{Opt}] - \Pr [ a_i(t-1) = 1 |\pi_i^{SPI}] \right|\nonumber\\
    &+5r_{\max} \cdot \left| \Pr [ a_i(t-1) = 1 |\pi_i^{Opt}] - \Pr [ a_i(t-1) = 1 |\pi_i^{SPI}] \right|,
\end{align}
where the inequality comes from Lemma \ref{lemma:4}.
Therefore, we have
\begin{align}
    &\mathbb{E}[\Delta V]\leq \sum_{t=1}^T\sum_{i=1}^N\sum_{j=1}^\rho\left| \mathbb{E} [ R_{i,j}^{\text{Opt}}(t) ] - \mathbb{E} [ R_{i,j}^{\text{SPI}}(t) ] \right|\nonumber\\
    &\leq \sum_{t=1}^T\sum_{i=1}^N\sum_{j=1}^\rho 5r_{\max} \cdot \left| \Pr [ a_{i,j}(t) = 1 |\pi_i^{Opt}] - \Pr [ a_{i,j}(t) = 1 |\pi_i^{SPI}]\right|.
\end{align}
Next, we introduce the Hoeffding inequality \cite{hoeffding1994probability}.
\begin{lemma}[Hoeffding inequality\cite{hoeffding1994probability}]\label{lemma:hoeffding}

Let \( Y_1, Y_2, \dots, Y_n \) be independent random variables such that \( Y_i \in [a_i, b_i] \). Let \( S_n = \sum_{i=1}^{n} Y_i \) be the sum of these random variables. Then for any \( t > 0 \), Hoeffding's inequality states:

\begin{align*}
   Pr\left( |S_n - \mathbb{E}[S_n]| \geq \epsilon \right) \leq 2 \exp\left( -\frac{2\epsilon^2}{\sum_{i=1}^{n} (b_i - a_i)^2} \right).
\end{align*}

\end{lemma}

Define $Y_j:=\left| \Pr [ a_{i,j}(t) = 1 |\pi_i^{Opt}] - \Pr [ a_{i,j}(t) = 1 |\pi_i^{SPI}]\right|$ and $S_\rho:=\sum_{j=1}^\rho Y_j$. Then according to Lemma \ref{lemma:hoeffding}, we have
\begin{align*}
    Pr\left( |S_\rho - \mathbb{E}[S_\rho]| \geq \epsilon \right) \leq 2 \exp\left( -\frac{2\epsilon^2}{\rho} \right),
\end{align*}
as $Y_j\in[0,1]$.
Therefore, with probability at least $1-\frac{1}{\rho}$ that
\begin{align}
\mathbb{E}[\Delta V]\leq 5r_{\text{max}}NT\sqrt{\frac{\ln 2\rho}{2\rho}}
\end{align}

\subsubsection{Total bound} Hence, we can bound $\sum_{i=1}^N X_i$ as
\begin{align}
    \sum_{i=1}^N X_i&\leq \mathbb{E}[\Delta V]+\sqrt{\frac{N\rho c^2\ln 2\rho}{2}}\nonumber\\
    &\leq 5r_{\text{max}}NT\sqrt{\frac{\ln 2\rho}{2\rho}}+\sqrt{\frac{N\rho c^2\ln 2\rho}{2}}\nonumber\\
    &\leq 5r_{\text{max}}NT\sqrt{\frac{\ln 2\rho}{2\rho}}+2r_{\text{max}}T\sqrt{\frac{N\rho \ln 2\rho}{2}}.
\end{align}

Therefore, we have the final result as
\begin{align}
    &\frac{1}{\rho N}\Big(J^{\pi^{Opt}}(\rho B, \rho N)- J^{\pi^{SPI}}(\rho B, \rho N)\Big)\nonumber\\
    &\leq 5r_{\text{max}}T\sqrt{\frac{\ln 2\rho}{2\rho^3}}+2r_{\text{max}}T\sqrt{\frac{ \ln 2\rho}{2N\rho}}.
\end{align}

\section{Experiment Details and Additional Results}

\subsection{Domain Details}\label{sec:C1}
In this subsection, we provide the details of the three domains considered in Section \ref{Sec:exp}.

\subsubsection{Continuous Positive Airway Pressure Therapy (CPAP)}
In the CPAP model with $S$ states from $1$ to $S$, if no intervention is taken (i.e., action  $a = 0$ ), the patient has a probability of 1 to move from a higher adherence level $s$ to a lower adherence level $s-1$. When intervention is applied, i.e., $a=1$, the patient may either move to a lower adherence level $s-1$ or higher adherence level $s+1$ with certain probabilities.
The state transition probabilities can be written as:

\begin{itemize}
    \item \textbf{Action 0}:
    $
    P_n(s, a=0, s-1) = 1, \forall n, s.
    $
    \item \textbf{Action 1}:
    $
    P(s, a=1, s-1) +
    P(s, a=1, s+1) = 1.
    $
\end{itemize}
For the above transitions, if $s$ is the boundary state $0$, taking action $0$ will make the arm stay at state $0$. If $s$ is the boundary state $S$, taking action $1$ will lead to $
    P(S, a=1, S-1) +
    P(S, a=1, S) = 1.
    $
For reward functions, a higher state provides a higher reward if the patient receives an intervention. We assume that the transition kernels and reward functions are heterogeneous across different types of patients.

\subsubsection{Mobile Healthcare for Maternal Health (MHMH)}

Consider two families of arms corresponding to
the reliable and greedy types of patients, each family has N types of arms. Also, there are three states per arm (start $s_s$, engaged $s_e$, and dropout $s_d$), and two actions (call$ 1$ and no call $0$).
Here are the details of transition matrix: all arms start in the start state $s_s$. Then with some fixed probability $\eta_{g,s},\eta_{r,s}$, it evolves to engaged state or dropout state. For greedy arms in engaged state, if will evolve to dropout state whatever the action is. For reliable arms in engaged state, it will stay engaged for sure if action is taken, and stay engaged for some fixed probability $\eta_{r,e}$ less than 1 if action isn't taken. For all arms in dropout state, it will either stay in dropout state or evolve to start state following some fixed probability $\eta_{g,d},\eta_{r,d}$. Table \ref{tab:formula_table} shows all transition probabilities in general.
We also consider the reward collected by action, which means we only get reward from one arm if we pull it in this round. For greedy arms, the reward in engaged state is $1$. For reliable arms, the reward in engaged state is some constant $C$ less than $1$. The rewards for other states are all $0$. The total budget is half of number of arms.

\begin{table}[t]
    \centering
    \begin{tabular}{|c|c|c|}
    \hline
    % First row: header for the second and third columns
    & \textbf{Greedy Arms} & \textbf{Reliable Arms} \\
    \hline
    % Second row: multirow for "Active" with formulas in the other cells
    \multirow{3}{*}{\textbf{Active}} &
    $P^{1}_{s_s,s_e}=1$ & $P^{1}_{s_s,s_e}=1$ \\
    & $P^{1}_{s_e,s_d}=1$ & $P^{1}_{s_e,s_e}=1$ \\
    & $P^{1}_{s_d,s_s}=\eta_{g,d}=1-P^{1}_{s_d,s_d}$ & $P^{1}_{s_d,s_s}=\eta_{r,d}=1-P^{1}_{s_d,s_d}$ \\
    \hline
    % Third row: multirow for "Passive" with formulas in the other cells
    \multirow{3}{*}{\textbf{Passive}} &
    $P^{0}_{s_s,s_e}=\eta_{g,s}=1-P^{0}_{s_s,s_d}$ & $P^{0}_{s_s,s_e}=\eta_{r,s}=1-P^{0}_{s_s,s_d}$ \\
    & $P^{0}_{s_e,s_d}=1$ & $P^{0}_{s_e,s_e}=\eta_{r,e}=1-P^{0}_{s_e,s_d}$ \\
    & $P^{0}_{s_d,s_s}=\eta_{g,d}=1-P^{0}_{s_d,s_d}$ & $P^{0}_{s_d,s_s}=\eta_{r,d}=1-P^{0}_{s_d,s_d}$ \\
    \hline
    \end{tabular}
    \caption{Transition Probabilities in Mobile Healthcare for Maternal Health Domain. Other transition probabilities not shown in the graph are all $0$.}
    \label{tab:formula_table}
\end{table}

\subsubsection{Enhrenfest Project}
Consider $N$ Enhrenfest projects with state spaces represented as $s = 0, \cdots, S$. In the active phase, the reward rate for an arm is $c \cdot s$, and the reward is zero during the passive phase. Transition dynamics in the active phase allow movement from state $s$ to $s-1$ with rate $\mu \cdot s$, whereas in the passive phase, the transition is from $s$ to $s+1$ with rate $\lambda \cdot (S-s)$. This model simulates the operation and recovery of an arm: active phases yield rewards as the arm expends energy, while passive phases involve no reward generation as the arm regains its state. Whittle \cite{whittle1988restless} provides a closed form for the Whittle index of these projects:
$$v(s) = \frac{c}{\mu \cdot S} (\mu \cdot s^2-\lambda \cdot (S-s)^2)$$
We implement a discretization with a time step of $\Delta_t = 0.01$. For each arm $i$, the parameter $c_i$ is drawn uniformly at random from intervals $(1, 10)$, and $\mu_i$, and $\lambda_i$ are drawn uniformly at random from $(0,10)$. Considering that only active arms generate rewards, a no-pull policy is ineffectual in this context.

\begin{table*}[h]
    \centering
    \scalebox{0.9}{
    \begin{tabularx}{\textwidth}{l>{\centering\arraybackslash}X>{\centering\arraybackslash}X>{\centering\arraybackslash}X}
    \toprule
    \multirow{2}{*}{\textbf{Policy}} & \multicolumn{3}{c}{\textbf{Enhrenfest Project}}  \\
    \cmidrule(lr){2-4}
    & {\textbf{$(10, 10, 3, 3, 10)$}}  & {\textbf{$(30, 5, 20, 10, 6)$}}  & {\textbf{$(20, 10, 6, 3, 10)$}}\\
    \midrule
    Upper Bound& \colorbox{white}{21.3} & \colorbox{white}{42.1} & \colorbox{white}{41.1} \\
    SPI & \colorbox{green}{$20.4\pm 0.5$} & \colorbox{green}{$41.3\pm 0.5$} & \colorbox{green}{$40.3\pm 0.9$} \\
    Mean Field & \colorbox{white}{$12.8\pm 0.3$} & \colorbox{white}{$27.6\pm 0.3$} & \colorbox{white}{$19.8\pm 0.4$} \\
    Finite Whittle & \colorbox{green}{$20.8\pm 0.3$} & \colorbox{green}{$41.8\pm 0.3$} & \colorbox{green}{$40.2\pm 0.6$} \\
    Infinite Whittle & \colorbox{green}{$21.0\pm 0.3$} & \colorbox{green}{$41.7\pm 0.3$} & \colorbox{green}{$40.5\pm 0.5$} \\
    Original Whittle & \colorbox{white}{$11.5\pm 0.2$} & \colorbox{white}{$28.3\pm 0.2$} & \colorbox{white}{$17.7\pm 0.3$} \\
    Q-Difference & \colorbox{green}{$20.9\pm 0.3$} & \colorbox{green}{$41.8\pm 0.3$} & \colorbox{green}{$40.7\pm 0.6$} \\
    Random & \colorbox{white}{$11.0\pm 0.2$} & \colorbox{white}{$14.9\pm 0.2$} & \colorbox{white}{$17.0\pm 0.3$} \\
    \bottomrule
    \end{tabularx}
    }
    \caption{We present the performance of all policies under \emph{Enhrenfest project}. We run each setting for $100$ simulations and present $95\%$ confidence interval. Each setting is denoted by the parameters (number of types $N$, number of states $S$, budget $K$, group size $\rho$, time horizon $T$). Optimal policies are highlighted in green.}
    \label{tab:enhrenfest}
    \vspace{-0.1in}
\end{table*}

\begin{figure*}[h]
\centering
\includegraphics[width=0.99\textwidth]{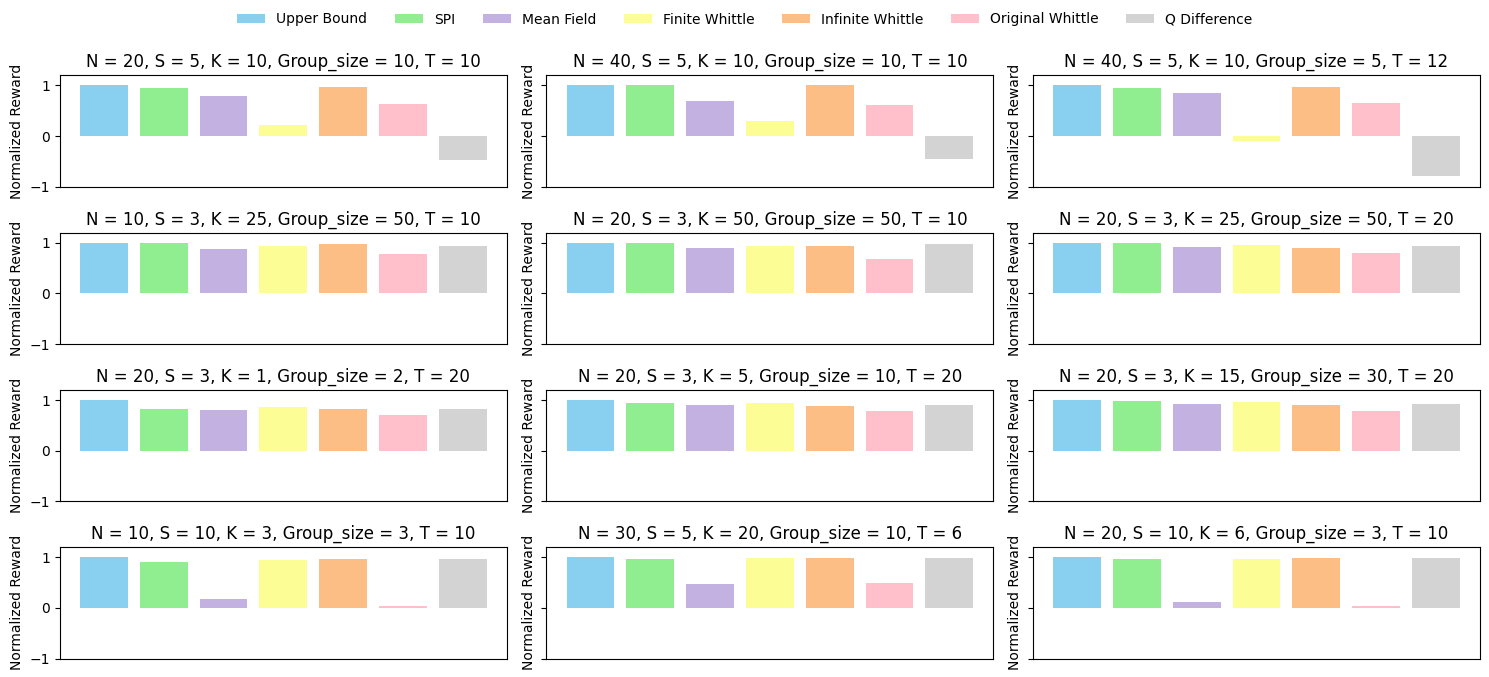}
\caption{We present the chart corresponding to Table \ref{tab:formula_table}, where the performance of all policies is normalized between the optimal upper bound and the random policy. Specifically, the normalized optimal upper bound is set to 1, while the performance of the random policy is set to 0, providing a clear comparison of the relative performance across all policies.}
\label{fig:table1}
\vspace{-0.1in}
\end{figure*}

\subsection{Additional Results}\label{sec:C2}

We show results of CPAP and MHMH domains in Table \ref{tab:formula_table} and Enhrenfest Project in Table \ref{tab:enhrenfest}. Figure \ref{fig:table1} is the chart corresponding to Table \ref{tab:formula_table} and \ref{tab:enhrenfest}, where the performance of all policies is normalized between the
optimal upper bound and the random policy. Specifically, the normalized optimal upper bound is set to 1, while the performance
of the random policy is set to 0, providing a clear comparison of the relative performance across all policies.

\subsubsection{Averaging over different MDP Instances}
In each setting of Table \ref{tab:formula_table}, we randomly assign a set of transition probabilities for each type of arm. To evaluate the robustness of our results across different MDP instances, we use random seeds from $0$ to $19$ to generate 20 MDP instances in one CPAP setting, running 10 simulations for each instance. The results are displayed in Figure \ref{fig: robustness}. As shown, the \textit{SPI} policy and the infinite Whittle index policy perform as the optimal ones, indicating that our method is robust across different MDP instances.

\begin{figure*}[h]
\centering
\includegraphics[width=0.6\textwidth]{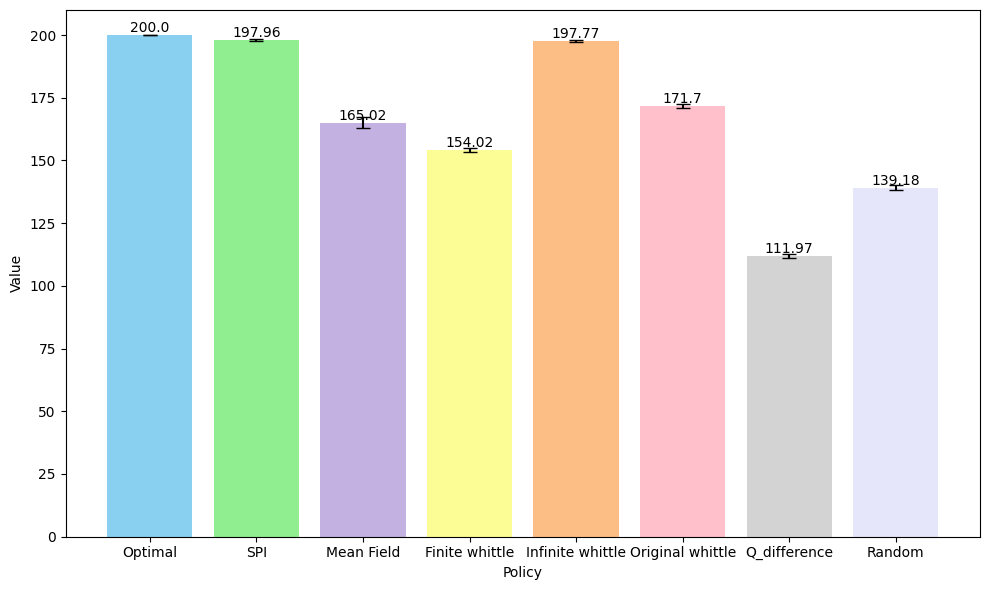}
\vspace{-0.1in}
\caption{Here we consider the setting $(N,S,K,\rho,T)= (20,5,10,10,10)$ in the Birth-Death Process(CPAP) domain. We pick $20$ different random seeds ($0$-$19$) and run $10$ simulations for MDP generated by each random seed. Here the chart shows the average performance and $95\%$ confidence interval of each policy. }
\label{fig: robustness}
\vspace{-0.13in}
\end{figure*}

\subsubsection{Robustness Check for Running Time}
\begin{figure*}[h]
\centering
\includegraphics[width=0.5\textwidth]{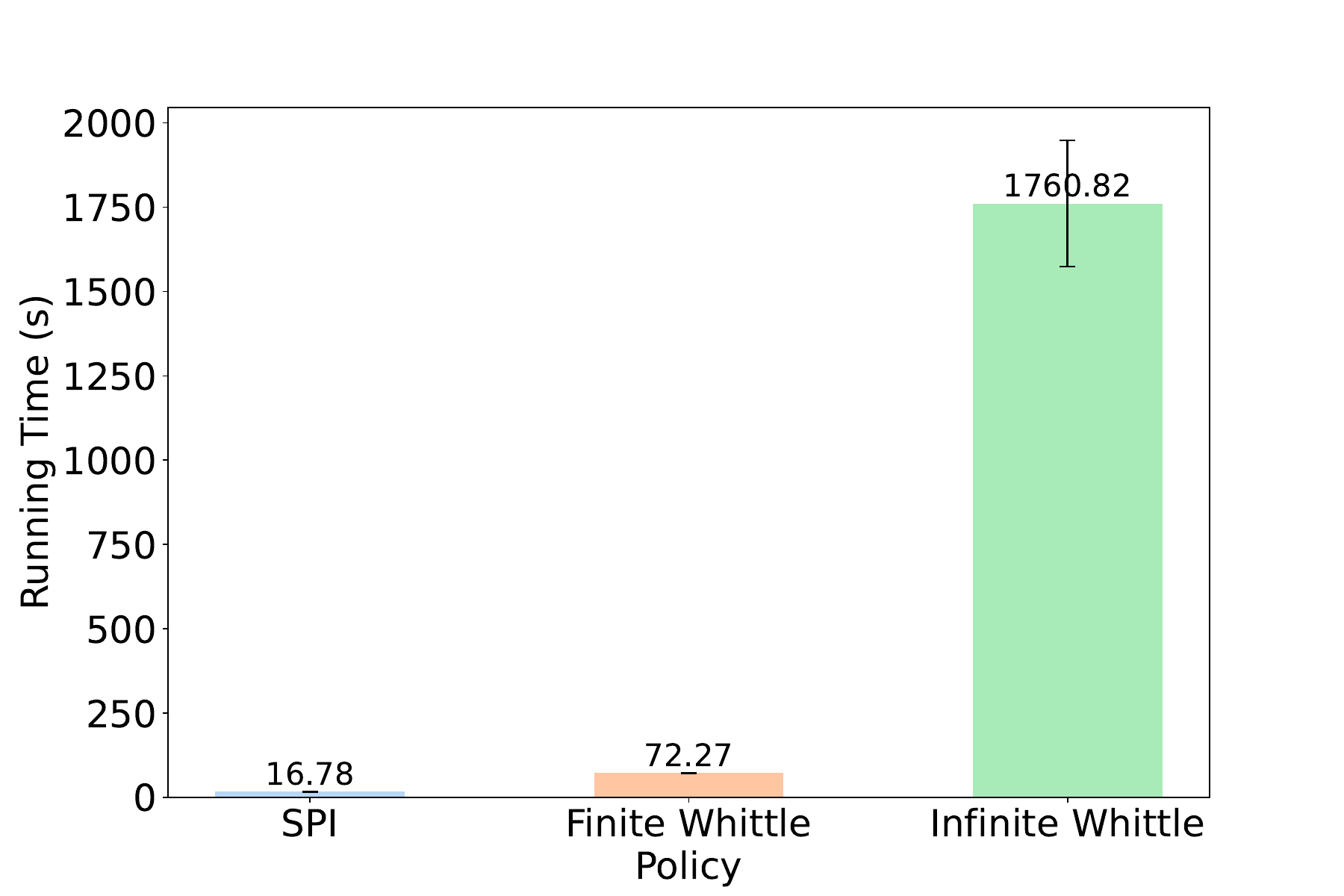}
\vspace{-0.1in}
\caption{We present the average running time of \texttt{SPI} policy, finite whittle policy, and infinite whittle policy in three different randomly generated MDPs with $(N,S,K,\rho,T) = (10, 10, 50, 50, 10)$.}
\label{fig:Random_run_time}
\vspace{-0.13in}
\end{figure*}

We randomly generate three MDP instances, and we take the average running time of each policy under different MDP instances. Figure \ref{fig:Random_run_time} serves as a robustness check, affirming our finding that the SPI index generally exhibits significantly lower running times than those observed with Whittle-index-based policies.

\subsubsection{Asymptotic Optimality}\label{subsubsec:asymptotic_optimality}
% \textbf{Asymptotic Optimality.}
\begin{figure*}[htbp]
\centering
\includegraphics[width=0.5\textwidth]{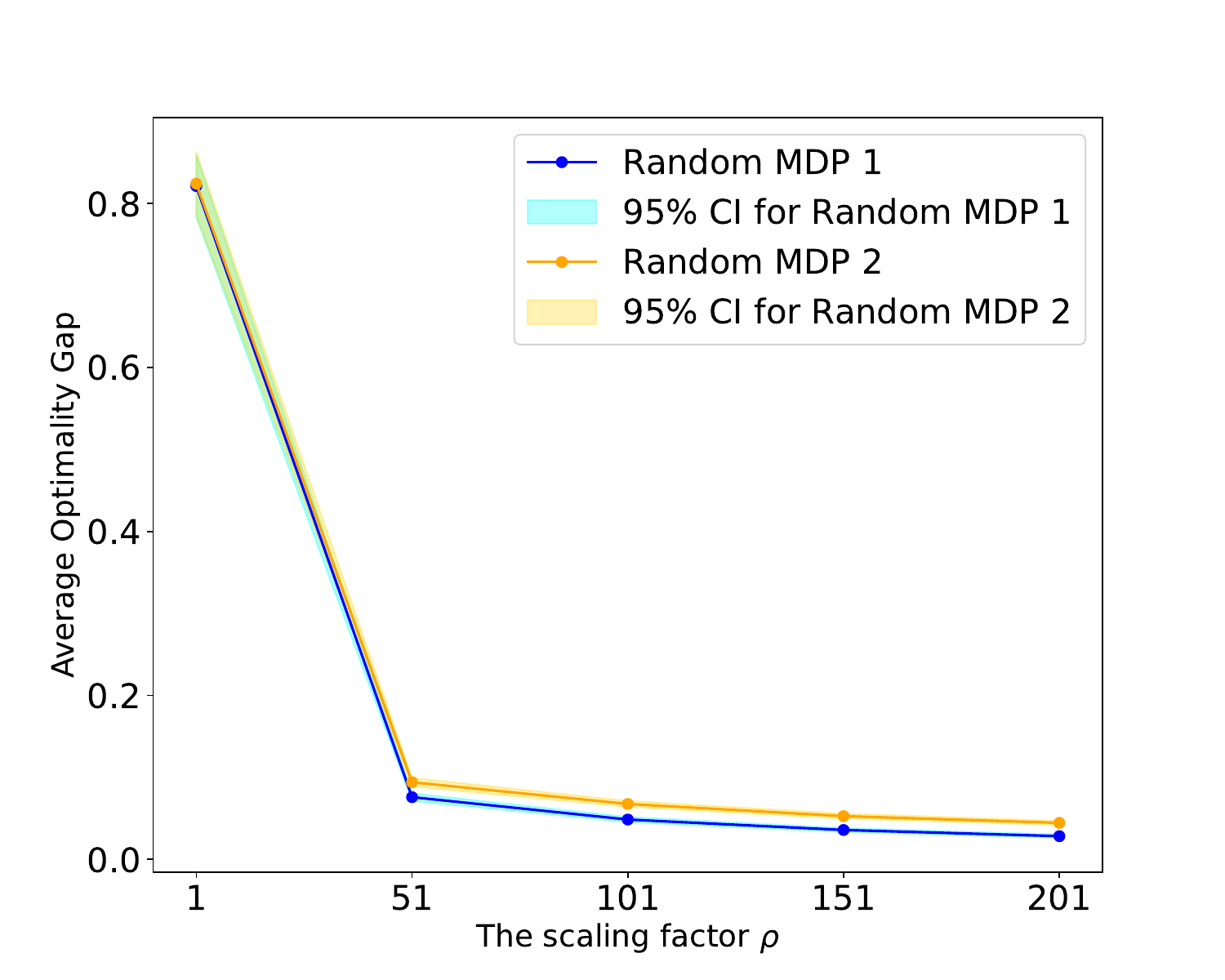}
\vspace{-0.2in}
\caption{We present the asymptotic optimality of the proposed \texttt{SPI} policy with $(N,S,K,\rho,T) = (20, 10, 30, 10, 6)$. }
\label{fig:asymptotic_optimality}
\end{figure*}
As indicated by Theorems \ref{thm:asym_opt} and \ref{thm:non-asymptotic}, the proposed \texttt{SPI} policy is asymptotically optimal as the number of total arms $\rho N$ goes large, where $\rho$ is the scaling factor of the system. At this end, we
 empirically demonstrate the asymptotic optimality of the \texttt{SPI} policy under two different MDP instances whose transition probabilities are randomly generated, as shown in Figure \ref{fig:asymptotic_optimality}.  We observe that as $\rho$ increases, the average optimality gap decreases,  converging towards zero. This result illustrates that as the system scales, the \texttt{SPI} policy becomes increasingly optimal, approaching the theoretical reward upper bound. The figure also includes 95\% confidence intervals for both MDP instances, confirming the reliability of the observed trend.

\end{document}